\documentclass{article}

\usepackage[utf8]{inputenc}
\usepackage{amsmath}
\usepackage{amsthm}
\usepackage{amssymb}
\usepackage{braket}
\usepackage{comment}
\usepackage{hyperref}
\usepackage{authblk}
\usepackage{cite}
\usepackage{bm}
\usepackage[inline]{enumitem}

\theoremstyle{plain}
\newtheorem{thm}{Theorem}
\newtheorem{lem}{Lemma}
\newtheorem{cor}{Corollary}
\newtheorem{pro}{Proposition}

\DeclareMathOperator{\tr}{tr}

\newcommand{\hil}{\mathcal{H}}
\newcommand{\linear}{\mathcal{L}}

\newcommand{\id}{\mathcal{I}}

\begin{document}
\title{Upper bounds on probabilities in channel measurements on qubit channels and their applications}
\author[1]{Taihei Kimoto\footnote{kimoto.taihei.22e@st.kyoto-u.ac.jp}}
\author[2]{Takayuki Miyadera\footnote{miyadera@mi.meijigakuin.ac.jp}}
\affil[1]{Department of Nuclear Engineering, Kyoto University, Kyoto daigaku-katsura\\
Nishikyo-ku, Kyoto 615-8540,
Japan}
\affil[2]{Institute for Mathematical Informatics, Meiji Gakuin University, Kanagawa 244-8539, Japan}
\date{}

\maketitle
\begin{abstract}
One of the fundamental tasks in quantum information processing is to measure the quantum channels.
Similar to measurements of quantum states, measurements of quantum channels are inherently stochastic, that is, quantum theory provides a formula to calculate the probability of obtaining an outcome.
The upper bound on each probability associated with the measurement outcome of the quantum channels is a fundamental and important quantity.
In this study, we derived the upper bounds of the probability in a channel measurement for specific classes of quantum channels.
We also present two applications for the upper bounds.
The first is the notion of convertibility considered by Alberti and Uhlmann and the second is the detection problem of a quantum channel.
These applications demonstrate the significance of the obtained upper bounds.
\end{abstract}

\section{Introduction}
A characteristic property of quantum theory is that it provides a formula for the probabilities in measurements.
In other words, the probability of obtaining an outcome can be calculated using the Born rule.
While quantum states are often regarded as the basic targets to be measured, quantum channels can also be measured.
Measurement of quantum channels is a fundamental task because channels are the key components of important protocols or algorithms in quantum information processing\cite{nielsen_chuang_2010}.
In Ref. \citen{PhysRevA.77.062112}, a framework for measuring quantum channels is formulated, and a formula for calculating the probabilities associated with such measurements is derived.
Obviously, an upper bound on probability is a fundamental and important quantity, as it limits the certainty of an outcome occurring.

Let us consider a game between Alice and Bob to illustrate the significance of an upper bound on a probability in a channel measurement.
Alice has a device for measuring channels and acts as a referee, while Bob has ability to prepare a channel and acts as a challenger.
The game proceeds as follows:
\begin{enumerate}
\item Alice chooses a measurement outcome and communicates to Bob the device that she has and the selected outcome.
\item Based on the information from Alice, Bob sends a channel to Alice.
\item Alice measures the received channel using the device.
\end{enumerate}
Bob wins the game if the measurement outcome matches the one specified by Alice in the first step.
Clearly, Bob's probability of winning is determined by the probability of obtaining the correct outcome, and the upper bound on this probability is the limit of his winning probability.
Furthermore, various classes of quantum channels exist\cite{Heinosaari_Ziman_2011}.
Because Bob may be able to prepare channels in a certain class, it is logical to consider such an upper bound for each class.
This study aims to derive an upper bound on the probability of channel measurement for a typical subset of quantum channels.
One of the related quantities is already known as the fully entangled fraction (FEF)\cite{PhysRevA.54.3824}.
Although the FEF is often used in the context of bipartite systems, through the so-called Choi--Jamiołkowski isomorphism\cite{CHOI1975285,JAMIOLKOWSKI1972275}, it can be utilized in situations where quantum channels are discussed.
In Ref. \citen{https://doi.org/10.48550/arxiv.2205.10108}, we derived an uncertainty relation for measurements of random unitary channels using the FEF.

As an application, we relate the derived upper bounds to the notion of convertibility discussed by Alberti and Uhlmann\cite{ALBERTI1980163}, examining this in the context of various classes of quantum channels.
Additionally, we address the problem of detecting the properties of quantum channels.
More specifically, we provide criteria verifying a given unital channel is not entanglement breaking.
Consequently, the upper bounds on the probabilities of channel measurements are not only fundamental quantities but also practical tools.

This paper is organized as follows.
Section \ref{sec:pre}, summarizes the preliminaries.
Section \ref{sec:est}, the main section, derives upper bounds for several classes of quantum channels.
Section \ref{sec:jc} and Section \ref{sec:id} present the applications of these upper bounds, indicating their significance.
Section \ref{sec:conc} provides the conclusions drawn from the study.
Appendix provides the proof of the lemma used in the main text.

\section{Preliminaries} \label{sec:pre}
In this section, we summarize the preliminary knowledge used in this study.

A quantum system is represented by a Hilbert space $\hil$.
In this study, all quantum systems are qubits unless otherwise specified.
The set of all linear operators on $\hil$ is denoted by $\linear(\hil)$.
A quantum state is expressed using a density operator, $\rho$ which is a linear operator that satisfies the conditions $\rho\geq0$ and $\tr\rho=1$.
Here, the first condition indicates that $\rho$ is a positive operator, and $\tr$ in the second condition denotes the trace.
An effect $E$ is a positive operator that satisfies the condition $E\leq I$, where $I$ denotes the identity operator.
Let us consider a measurement of quantum states with a finite outcome set $\Omega$.
This measurement is expressed by a positive-operator-valued measure (POVM), which is a family of effects $\{E_{m}\}_{m\in\Omega}$ normalized such that $\sum_{m\in\Omega}E_{m}=I$.
For a quantum state $\rho$, the probability of obtaining an outcome $m$ is given by the Born rule
\begin{equation}
p_{m}(\rho)=\tr[E_{m}\rho]. \label{eq:born}
\end{equation}
A quantum channel is described as a linear, completely positive, and trace preserving map $\Psi:\linear(\hil)\rightarrow\linear(\hil)$.
Owing to the Choi--Jamiołkowski isomorphism, the completely positive condition of $\Psi$ is equivalent to
\begin{equation}
J_{\Psi}\geq 0. \label{eq:cj_op}
\end{equation}
Here, the left-hand side is given by
\begin{equation}
J_{\Psi}=(\id\otimes\Psi)P'_{+},
\end{equation}
where $\id$ represents the identity map, and $P'_{+}$ denotes the unnormalized maximally entangled state defined as
\begin{equation}
P'_{+}=\ket{\psi'_{+}}\bra{\psi'_{+}}=\sum^{1}_{i,j=0}\ket{i}\bra{j}\otimes\ket{i}\bra{j}.
\end{equation}

Next, we present a method for measuring quantum channels, and therefore, we employ the framework formulated in Ref. \citen{PhysRevA.77.062112}.
Let us consider a quantum channel $\Psi$ on $\linear(\hil)$.
We can observe how $\Psi$ transforms a quantum state.
Hence, a simple approach to measure $\Psi$ is inputting a known state $\rho$ and conducting a known POVM measurement $\{E_{m}\}_{m\in\Omega}$ on the output state.
In general, $\rho$ can be chosen as a state correlated with an additional system $\hil_{anc}$, termed an ancilla.
In this case, $\{E_{m}\}_{m\in\Omega}$ is applied to the system composed of the output and the ancilla.
While $\hil$ represents a qubit, we assume the ancilla to be an arbitrary finite dimensional system.
Thus, a measurement of quantum channels is defined by the input state $\rho$ on $\hil_{anc}\otimes\hil$ and the POVM $\{E_{m}\}_{m\in\Omega}$ on $\hil_{anc}\otimes\hil$.
For simplicity, we denote this measurement as the couple $\mathcal{T}=(\rho, \{E_{m}\}_{m\in\Omega})$.
From (\ref{eq:born}), the probability of obtaining an outcome $m$ is given by
\begin{equation}
p_{m}(\Psi)=\tr\left[E_{m}(\id\otimes\Psi)\rho\right].
\end{equation}
In Ref. \citen{PhysRevA.77.062112}, it was demonstrated that the right-hand side can be rewritten as
\begin{equation}
p_{m}(\Psi)
=\tr\left[S_{m}J_{\Psi}\right].
\end{equation}
Here, $S_{m}$ represents a positive operator determined by $\mathcal{T}$, and termed the process-channel effect.
More specifically, there exists a completely positive map $\Upsilon_{\rho}$ such that $\rho=(\Upsilon_{\rho}\otimes\id)(P'_{+})$, and $S_{m}$ is defined as $S_{m}=(\Upsilon^{*}_{\rho}\otimes\id)(E_{m})$, where $\Upsilon^{*}_{\rho}$ denotes the dual map of $\Upsilon_{\rho}$ with respect to the Hilbert--Schmidt inner product.
The family $\{S_{m}\}_{m\in\Omega}$ is referred to as process POVM (PPOVM).
Unlike POVMs, $\{S_{m}\}_{m\in\Omega}$ satisfies the normalization condition
\begin{equation}
\sum_{m\in\Omega}S_{m}=(\tr_{anc}\rho)^{T}\otimes I.
\end{equation}

Let us illustrate two typical examples of PPOVMs.
The first pertains to the measurement of quantum channels without involving ancilla systems.
This measurement type is termed ancilla-free and is defined by $(\rho,\{E_{m}\}_{m\in\Omega})$, where $\rho$ is a density operator on $\hil$ and $\{E_{m}\}_{m\in\Omega}$ is a POVM on $\hil$.
In this case, the PPOVM is given by
\begin{equation}
\{\rho^{T}\otimes E_{m}\}_{m\in\Omega}.
\end{equation}

The second example corresponds to a measurement utilizing the maximally entangled state $P_{+}=P'_{+}/2$ as the input state, specified by $(P_{+},\{E_{m}\}_{m\in\Omega})$, where $\{E_{m}\}_{m\in\Omega}$ is a POVM on $\hil_{anc}\otimes\hil$.
The PPOVM is given by
\begin{equation}
\left\{\frac{1}{2}E_{m}\right\}_{m\in\Omega}.
\end{equation}

We frequently employ these PPOVMs as examples in this study.

\section{Estimation of probabilities in channel measurements} \label{sec:est}
Let us consider a nonzero process-channel effect $S$ such that $S\leq\rho^{T}_{0}\otimes I$ where $\rho_{0}$ is a density operator.
The probability associated with $S$ is expressed as
\begin{equation}
p(\Psi)=\tr\left[SJ_{\Psi}\right], \label{eq:ch_born}
\end{equation}
where $\Psi$ belongs to a quantum channel class $\mathcal{X}$.
As mentioned in the introduction, the primary objective of this study is to estimate this value.
By defining $\tau_{S}=S/\left(\tr{S}\right)$, we can simplify (\ref{eq:ch_born}) as
\begin{equation}
p(\Psi)=\left(\tr{S}\right)\tr\left[\tau_{S}J_{\Psi}\right].
\end{equation}
Hence, for an arbitrary density operator $\tau$ on $\hil\otimes\hil$, let us examine the quantity
\begin{equation}
\tr[\tau J_{\Psi}],
\end{equation}
and its maximum
\begin{equation}
C(\tau,\mathcal{X})=\max_{\Psi'\in\mathcal{X}}\tr[\tau J_{\Psi'}]. \label{eq:def_ub}
\end{equation}
It is evident that the inequality
\begin{equation}
p(\Psi)\leq\left(\tr{S}\right)C(\tau_{S},\mathcal{X})
\end{equation}
holds for any $\Psi\in\mathcal{X}$; therefore, we need to compute or estimate (\ref{eq:def_ub}).
In the following subsections, we undertake this task for typical classes of quantum channels.

\subsection{Depolarizing channels}
A depolarizing channel is characterized by a quantum state $\rho$ on $\hil$, and it is defined by
\begin{equation}
\Psi_{\rho}(A)=(\tr A)\rho,
\end{equation}
for a linear operator $A$ in $\linear(\hil)$.
We denote the set of depolarizing channels by
\begin{equation}
\mathcal{D}=\Set{\Psi_{\rho}|\text{$\rho$ is a quantum state on $\hil$.}},
\end{equation}
and we derived the following proposition.
\begin{pro} \label{pro:gdc}
In the case of $\mathcal{D}$, (\ref{eq:def_ub}) becomes
\begin{equation}
C(\tau,\mathcal{D})=\frac{1}{2}\left(1+\|\bm{b}(\tau_{2})\|\right), \label{eq:bound_gd}
\end{equation}
where $\tau_{2}$ is defined as $\tau_{2}=\tr_{1}\tau$.
The vector $\bm{b}(\cdot)$ denotes the Bloch vector, an element of the Euclidean space $\mathbb{R}^{3}$ whose entry is given by
\begin{equation}
b(\cdot)_{i}=\tr[\sigma_{i}(\cdot)], \label{eq:def_bloch}
\end{equation}
where $\sigma_{i}$ is the $i$th Pauli matrix.
The norm $\|\cdot\|$ represents the Euclidean norm, defined as
\begin{equation}
\|\bm{x}\|=\sqrt{x^{2}_{1}+x^{2}_{2}+x^{2}_{3}} \label{eq:def_euc_norm}
\end{equation}
for all $\bm{x}\in\mathbb{R}^{3}$. 
\end{pro}
\begin{proof}
For a depolarizing channel $\Psi_{\rho}$, it is simple to verify that
\begin{align}
J_{\Psi_{\rho}}
=I\otimes\rho.
\end{align}
Hence, the maximum is computed as
\begin{align}
C(\tau,\mathcal{D})
&=\max_{\Psi\in\mathcal{D}}\tr[\tau J_{\Psi}] \\
&=\max_{\rho}\tr[\tau(I\otimes\rho)] \\
&=\max_{\rho}\tr[\tau_{2}\rho] \\
&=\max_{\rho}\tr\left[\tau_{2}\frac{1}{2}\left(I+\sum^{3}_{i=1}b(\rho)_{i}\sigma_{i}\right)\right] \\
&=\frac{1}{2}\left(1+\max_{\rho}\left<\bm{b}(\rho),\bm{b}(\tau_{2})\right>\right) \\
&=\frac{1}{2}\left(1+\|\bm{b}(\tau_{2})\|\right),
\end{align}
where $b(\cdot)_{i}$ represents the $i$th element of the Bloch vector and $\left<\cdot,\cdot\right>$ denotes the Euclidean inner product.
The last expression is obtained by applying the Cauchy--Schwarz inequality.
Thus, the maximum is attained if the state is
\begin{equation}
\rho=\frac{1}{2}\left(I+\sum^{3}_{i=1}\frac{b(\tau_{2})_{i}}{\|\bm{b}(\tau_{2})\|}\sigma_{i}\right),
\end{equation}
which is normalized $\tau_{2}$ in terms of the length of the Bloch vector.
\end{proof}

\subsection{Unital channels}
Although some parts of this subsection were previously discussed in Ref. \citen{https://doi.org/10.48550/arxiv.2205.10108}, we include them here for completeness.
A quantum channel $\Psi$ is termed unital if it satisfies the condition
\begin{equation}
\Psi(I)=I.
\end{equation}
We denote the set of all unital channels by $\mathcal{B}$, as they are often referred to as bistochastic.
Unitary channels are typical examples of unital channels.
Additionally, it is evident that a mixture of unitary channel is unital.
Mixtures of unitary channels are known as random unitary channels, and we denote the set of all random unitary channels as
\begin{equation}
\mathcal{R}=\Set{\sum_{x}p_{x}U_{x}(\cdot)U^{\dagger}_{x}},
\end{equation}
where $\{p_{x}\}$ is a probability distribution, and $\{U_{x}\}$ is a family of unitary operators.
Interestingly, for qubits, the equality $\mathcal{B}=\mathcal{R}$ holds\cite{LANDAU1993107}.
Thus, we can identify unital channels with random unitary channels.
Let us consider an arbitrary random unitary channel $\Psi$ expressed as
\begin{equation}
\Psi(A)=\sum_{x}p_{x}U_{x}AU^{\dagger}_{x}.
\end{equation}
It is evident that
\begin{align}
J_{\Psi}
=\sum_{x}p_{x}J_{\Psi_{x}},
\end{align}
where $J_{\Psi_{x}}$ is defined as $J_{\Psi_{x}}=(I\otimes U_{x})P'_{+}(I\otimes U^{\dagger}_{x})$.
Consequently, the maximum is achieved by a unitary channel because the inequality
\begin{align}
\tr[\tau J_{\Psi}]
&=\sum_{x}p_{x}\tr[J_{\Psi_{x}}] \\
&\leq\sum_{x}p_{x}\max_{\Phi\in\mathcal{U}}\tr[J_{\Phi}] \\
&=\max_{\Phi\in\mathcal{U}}\tr[J_{\Phi}]
\end{align}
holds.
Here, we represent the set of unitary channels by $\mathcal{U}$.
The maximum can be rewritten as
\begin{align}
C(\tau,\mathcal{R})
&=\max_{\Psi\in\mathcal{R}}\tr[\tau J_{\Psi}] \\
&=\max_{\Phi\in\mathcal{U}}\tr[\tau J_{\Phi}] \\
&=2\max_{U}\bra{\psi_{+}}(I\otimes U^{\dagger})\tau(I\otimes U)\ket{\psi_{+}}, \label{eq:c_r_fef}
\end{align}
where the maximization in the third equality is over all unitary operators, and $\ket{\psi_{+}}$ is the normalized maximally entangled state expressed as $\ket{\psi_{+}}=\ket{\psi'_{+}}/\sqrt{2}$.
The last expression is known as the FEF\cite{PhysRevA.54.3824} of $\tau$ defined by
\begin{equation}
f(\tau)=\max_{U}\bra{\psi_{+}}(I\otimes U^{\dagger})\tau(I\otimes U)\ket{\psi_{+}}.
\end{equation}
Fortunately, the analytical solution for the FEF of bipartite qubit states is well-established in several forms.
Especially, the following expression derived in Ref. \citen{PhysRevA.78.032332} provides an explicit formula.
\begin{thm}
Let $\sigma_{1}$, $\sigma_{2}$, and $\sigma_{3}$ be the Pauli matrices.
The FEF of $\tau$ is provided by
\begin{equation}
f(\tau)=\frac{1}{4}\left(1+\left\|N(\tau)^{T}N(P_{+})\right\|_{KF}\right), \label{eq:fef_sol}
\end{equation}
where $N(\cdot)$ is the so-called correlation matrix expressed as
\begin{equation}
N(\cdot)_{ij}=\tr[(\cdot)(\sigma_{i}\otimes\sigma_{j})]. \label{eq:def_corr_mat}
\end{equation}
The symbol $^{T}$ represents the transpose.
$\|\cdot\|_{KF}$ is the Ky Fan norm defined by
\begin{equation}
\|X\|_{KF}=\tr\sqrt{XX^{T}} \label{eq:def_kf_norm}
\end{equation}
for an arbitrary real matrix $X$.
\end{thm}
While (\ref{eq:fef_sol}) provides us a precise formula, it can be simplified further.
To do this, it is useful to express $P'_{+}$ as
\begin{align}
P'_{+}
=\frac{1}{2}\left(\sigma_{0}\otimes\sigma_{0}+\sigma_{1}\otimes\sigma_{1}-\sigma_{2}\otimes\sigma_{2}+\sigma_{3}\otimes\sigma_{3}\right). \label{eq:P'_dec}
\end{align}
From this equality, $P_{+}$ is expanded as
\begin{equation}
P_{+}=\frac{1}{4}\left(\sigma_{0}\otimes\sigma_{0}+\sigma_{1}\otimes\sigma_{1}-\sigma_{2}\otimes\sigma_{2}+\sigma_{3}\otimes\sigma_{3}\right),
\end{equation}
and its correlation matrix is computed as
\begin{equation}
N(P_{+})=
\begin{pmatrix}
1 & 0 & 0 \\
0 & -1 & 0 \\
0 & 0 & 1 \\
\end{pmatrix}
. \label{eq:P+_mat}
\end{equation}
Consequently, $N(P_{+})N(P_{+})^{T}=I$, and the second term of (\ref{eq:fef_sol}) becomes
\begin{align}
\left\|N(\tau)^{T}N(P_{+})\right\|_{KF}
=\left\|N(\tau)\right\|_{KF}.
\end{align}
From this equality, we can rewrite (\ref{eq:fef_sol}) as
\begin{equation}
f(\tau)=\frac{1}{4}\left(1+\left\|N(\tau)\right\|_{KF}\right). \label{eq:fef_sol_v2}
\end{equation}
By substituting (\ref{eq:fef_sol_v2}) into (\ref{eq:c_r_fef}), we obtain the following proposition.
\begin{pro} \label{eq:pro_unital}
In the case of $\mathcal{R}$, (\ref{eq:def_ub}) becomes
\begin{align}
C(\tau,\mathcal{R})
&=\frac{1}{2}\left(1+\left\|N(\tau)\right\|_{KF}\right), \label{eq:bound_unital}
\end{align}
where $N(\cdot)$ is the correlation matrix defined by (\ref{eq:def_corr_mat}), and $\left\|\cdot\right\|_{KF}$ is the Ky Fan norm defined by (\ref{eq:def_kf_norm}).
\end{pro}

\subsection{Unital entanglement breaking channels}
A quantum channel $\Psi$ is termed entanglement breaking channel if $(\id\otimes\Psi)(\rho)$ becomes separable for an arbitrary quantum state $\rho$ on $\hil\otimes\hil$\cite{doi:10.1142/S0129055X03001709}.
In this subsection, we focus on the set of all unital entanglement breaking channels, which we denote as $\mathcal{UE}$.
In Ref. \citen{doi:10.1142/S0129055X03001710}, the entanglement breaking channels on a qubit are analyzed in detail, and a necessary and sufficient condition for unital entanglement breaking channels is derived.
To introduce this condition, suppose $\Psi$ is unital.
The set $\{I,\sigma_{1},\sigma_{2},\sigma_{3}\}$ comprises a base of the vector space $\linear(\hil)$.
Thus, $\Psi$ is represented by a matrix $T$ whose entry is given by
\begin{equation}
T_{\mu\nu}=\frac{1}{2}\tr[\sigma_{\mu}\Psi(\sigma_{\nu})],
\end{equation}
where $\mu,\nu\in\{0,1,2,3\}$, and $\sigma_{0}=I$.
In Ref. \citen{904522}, it was demonstrated that $\Psi$ can be decomposed as
\begin{equation}
\Psi=\Psi_{V}\circ\Lambda\circ\Psi_{W}, \label{eq:ueb_dec}
\end{equation}
where $\Psi_{V}$ and $\Psi_{W}$ are unitary channels implemented by unitary operators $V$ and $W$, respectively, and $\Lambda$ is a quantum channel with the following matrix representation:
\begin{equation}
\begin{pmatrix}
1 & 0 & 0 & 0 \\
0 & \lambda_{1} & 0 & 0 \\
0 & 0 & \lambda_{2} & 0 \\
0 & 0 & 0 & \lambda_{3} \\
\end{pmatrix}
.
\end{equation}
The necessary and sufficient condition that was derived in Ref. \citen{doi:10.1142/S0129055X03001710} is as follows.
\begin{thm} \label{thm:ue_iff}
A unital channel $\Psi$ is entanglement breaking if and only if
\begin{equation}
\sum^{3}_{i=1}|\lambda_{i}|\leq1.
\end{equation}
\end{thm}
We use this condition to derive the following proposition:
\begin{pro} \label{pro:ueb}
In the case of $\mathcal{UE}$, (\ref{eq:def_ub}) is bounded from above as
\begin{equation}
C(\tau,\mathcal{UE})\leq\frac{1}{2}\left(1+\|N(\tau)\|\right). \label{eq:bound_ue}
\end{equation}
The norm $\|\cdot\|$ signifies the matrix norm induced by the Euclidean norm on $\mathbb{R}^{3}$; i.e., it is defined by
\begin{equation}
\|A\|=\max_{\bm{x},\|\bm{x}\|=1}\|A\bm{x}\| \label{eq:def_mat_norm}
\end{equation}
for all matrices $A$ on $\mathbb{R}^{3}$.
\end{pro}
\begin{proof}
Let $\Psi$ be an arbitrary unital entanglement breaking channel.
Based on (\ref{eq:ueb_dec}), it follows that
\begin{align}
J_{\Psi}
=(W^{T}\otimes V)J_{\Lambda}(W^{T}\otimes V)^{\dagger}.
\end{align}
Thus, the function to be maximized becomes
\begin{align}
\tr\left[\tau J_{\Psi}\right]
=\tr\left[\tau'J_{\Lambda}\right],
\end{align}
where $\tau'=(W^{T}\otimes V)^{\dagger}\tau(W^{T}\otimes V)$.
From (\ref{eq:P'_dec}), it holds that
\begin{align}
J_{\Lambda}
=\frac{1}{2}\left(\sigma_{0}\otimes\sigma_{0}+\sum^{3}_{i=1}(-1)^{i+1}\lambda_{i}\sigma_{i}\otimes\sigma_{i}\right).
\end{align}
Using this expression, we can estimate the function as
\begin{align}
\tr[\tau'J_{\Lambda}]
&=\frac{1}{2}\left(1+\sum^{3}_{i=1}(-1)^{i+1}\lambda_{i}\tr[\tau'(\sigma_{i}\otimes\sigma_{i})]\right) \nonumber \\
&=\left|\frac{1}{2}\left(1+\sum^{3}_{i=1}(-1)^{i+1}\lambda_{i}\tr[\tau'(\sigma_{i}\otimes\sigma_{i})]\right)\right| \nonumber \\
&\leq\frac{1}{2}\left(1+\sum^{3}_{i=1}|\lambda_{i}|\left|\tr[\tau'(\sigma_{i}\otimes\sigma_{i})]\right|\right) \nonumber \\
&\leq\frac{1}{2}\left(1+\sum^{3}_{i=1}|\lambda_{i}|\max_{i}\left|\tr[\tau'(\sigma_{i}\otimes\sigma_{i})]\right|\right) \nonumber \\
&\leq\frac{1}{2}\left(1+\left(\sum^{3}_{i=1}|\lambda_{i}|\right)\left(\max_{i}\left|\tr[\tau'(\sigma_{i}\otimes\sigma_{i})]\right|\right)\right) \nonumber \\
&\leq\frac{1}{2}\left(1+\max_{i}\left|\tr[\tau'(\sigma_{i}\otimes\sigma_{i})]\right|\right), \label{eq:ueb_est}
\end{align}
where the final inequality follows from the necessary and sufficient condition.
Based on the definition of $\tau'$, we can compute the second term of the last expression as
\begin{align}
\max_{i}\left|\tr[\tau'(\sigma_{i}\otimes\sigma_{i})]\right|
=\max_{i}\left|\tr[\tau(W^{T}\sigma_{i}(W^{T})^{\dagger}\otimes V\sigma_{i}V^{\dagger})]\right|. \label{eq:sec_term}
\end{align}
For further analysis, we focus on $W^{T}\sigma_{i}(W^{T})^{\dagger}$ and $V\sigma_{i}V^{\dagger}$.
From $\tr[V\sigma_{i}V^{\dagger})]=0$, we can expand $V\sigma_{i}V^{\dagger}$ as
\begin{align}
V\sigma_{i}V^{\dagger}
=\sum^{3}_{l=1}O(V)_{li}\sigma_{l}, \label{eq:V_exp}
\end{align}
where $O(V)_{li}=(1/2)\tr[\sigma_{l}V\sigma_{i}V^{\dagger}]$.
From
\begin{align}
(O(V)^{T}O(V))_{ij}
=\delta_{ij}, \label{eq:cof_is_ortho}
\end{align}
it is evident that $O(V)$ is an orthonormal matrix.
Likewise, $W^{T}\sigma_{i}(W^{T})^{\dagger}$ is expanded as
\begin{align}
W^{T}\sigma_{i}(W^{T})^{\dagger}=\sum^{3}_{k=1}O(W^{T})_{ki}\sigma_{k} \label{eq:W_exp}
\end{align}
where $O(W^{T})_{ki}=(1/2)\tr[\sigma_{k}W^{T}\sigma_{i}(W^{T})^{\dagger}]$ , and $O(W^{T})_{ki}$ is also an orthonormal matrix.
By substituting (\ref{eq:V_exp}) and (\ref{eq:W_exp}) into the right-hand side of (\ref{eq:sec_term}), we derive the equality
\begin{align}
\max_{i}\left|\tr[\tau'(\sigma_{i}\otimes\sigma_{i})]\right|
=\max_{i}\left|\sum^{3}_{k,l=1}O(W^{T})^{T}_{ik}\tr[\tau(\sigma_{k}\otimes\sigma_{l})]O(V)_{li}\right|. \label{eq:est_by_orth}
\end{align}
Note that $(O(W^{T})^{T}_{ik})^{3}_{k=1}$ is the transposed $i$th column vector of $O(W^{T})$, and $(O(V)_{li})^{3}_{l=1}$ is the $i$th column vector of $O(V)$.
Thus, by rewriting the last expression of (\ref{eq:est_by_orth}), we obtain the equality
\begin{align}
\max_{i}\left|\tr[\tau'(\sigma_{i}\otimes\sigma_{i})]\right|
=\max_{i}\left|\left<\mathbf{o}(W^{T})_{i},N(\tau)\mathbf{o}(V)_{i}\right>\right|, \label{eq:est_by_col}
\end{align}
where $\mathbf{o}(\cdot)_{i}$ denotes the $i$th column vector of $O(\cdot)$, and $\left<\cdot,\cdot\right>$ represents the Euclidean inner product defined as $\left<\mathbf{a},\mathbf{b}\right>=\sum^{3}_{k=1}a_{k}b_{k}$ for arbitrary vectors $\mathbf{a}$ and $\mathbf{b}$ in $\mathbb{R}^{3}$.
From the last expression of (\ref{eq:est_by_col}), we obtain the estimation
\begin{align}
\max_{i}\left|\tr[\tau'(\sigma_{i}\otimes\sigma_{i})]\right|
&=\max_{i}\left|\left<\mathbf{o}(W^{T})_{i},N(\tau)\mathbf{o}(V)_{i}\right>\right| \nonumber \\
&\leq\max_{i}\left\|N(\tau)\mathbf{o}(V)_{i}\right\| \nonumber \\
&\leq\left\|N(\tau)\right\|, \label{eq:sec_term_last_est}
\end{align}
where we applied the Cauchy–Schwarz inequality to derive the first inequality, used the definition of the matrix norm $\|N(\tau)\|=\max_{\mathbf{x},\|\mathbf{x}\|=1}\|N(\tau)\mathbf{x}\|$ to derive the second inequality, and utilized the fact that the Euclidean norm of the column vectors of an orthonormal matrix is unity.
By evaluating (\ref{eq:ueb_est}) with (\ref{eq:sec_term_last_est}), we obtain the inequality
\begin{align}
\tr[\tau'J_{\Lambda}]
\leq\frac{1}{2}\left(1+\left\|N(\tau)\right\|\right).
\end{align}
As $\Psi$ is arbitrary, it holds that
\begin{align}
C(\tau,\mathcal{UE})
&=\tr[\tau J_{\Phi_{0}}] \\
&\leq\frac{1}{2}\left(1+\left\|N(\tau)\right\|\right),
\end{align}
where $\Phi_{0}$ is defined as the unital entanglement breaking channel which achieves the maximum.
\end{proof}

\subsection{General entanglement breaking channels}
In this subsection, we address general entanglement breaking channels.
We denote the set of these channels as $\mathcal{GE}$.
To estimate $C(\tau,\mathcal{GE})$, it is sufficient to consider the extreme points of $\mathcal{GE}$ as the function to be maximized is affine regarding quantum channels.
For quantum channels on a qubit, the extreme points of $\mathcal{GE}$ are extreme classical-quantum (CQ) channels \cite{doi:10.1142/S0129055X03001709,doi:10.1142/S0129055X03001710}.
For a linear operator $A$ on $\hil$, an arbitrary extreme CQ channel $\Psi$ can be expressed as
\begin{equation}
\Psi(A)=\ket{\psi}\bra{x_{0}}A\ket{x_{0}}\bra{\psi}+\ket{\phi}\bra{x_{1}}A\ket{x_{1}}\bra{\phi},
\end{equation}
where $\ket{\psi}$ and $\ket{\phi}$ are quantum states, and $\{\ket{x_{0}},\ket{x_{1}}\}$ is an orthonormal base.
Based on these facts, we present the following proposition:
\begin{pro} \label{pro:geb}
In the case of $\mathcal{GE}$, (\ref{eq:def_ub}) is bounded from above as
\begin{equation}
C(\tau,\mathcal{GE})\leq\frac{1}{2}\left(1+\sqrt{\|\bm{b}(\tau_{2})\|^{2}+\|N(\tau)\|^{2}}\right), \label{eq:bound_ge}
\end{equation}
where $\bm{b}(\cdot)$ is the Bloch vector defined in (\ref{eq:def_bloch}), and $\|\bm{b}(\cdot)\|$ is the Euclidean norm defined in (\ref{eq:def_euc_norm}).
The matrix $N(\cdot)$ is the correlation matrix defined in (\ref{eq:def_corr_mat}), and $\|N(\cdot)\|$ is the matrix norm defined in (\ref{eq:def_mat_norm}).
\end{pro}
\begin{proof}
An entanglement breaking channel is expressed as
\begin{equation}
\Psi(A)=\ket{\psi}\bra{x_{0}}A\ket{x_{0}}\bra{\psi}+\ket{\phi}\bra{x_{1}}A\ket{x_{1}}\bra{\phi}
\end{equation}
for a linear operator $A$.
Here, we assume $0<\left|\braket{\psi|\phi}\right|<1$ because if $\left|\braket{\psi|\phi}\right|$ equals zero $\Psi$ becomes unital, and if $\left|\braket{\psi|\phi}\right|=1$, $\Psi$ is a depolarizing channel.
Both cases have been already considered in the previous subsections.
To represent $\Psi$ as a matrix, we define the unitary operator
\begin{align}
V=\ket{\psi_{+}}\bra{0}+\ket{\psi_{-}}\bra{1},
\end{align}
where $\ket{\psi_{\pm}}$ is defined by
\begin{equation}
\ket{\psi_{\pm}}=\frac{1}{\sqrt{2\left(1\pm\left|\braket{\psi|\phi}\right|\right)}}\left(\ket{\psi}\pm\frac{\braket{\phi|\psi}}{\left|\braket{\psi|\phi}\right|}\ket{\phi}\right).
\end{equation}
From the equalities
\begin{align}
\braket{\psi|\psi_{\pm}}
=\sqrt{\frac{1\pm\left|\braket{\psi|\phi}\right|}{2}}
\end{align}
and
\begin{align}
\braket{\phi|\psi_{\pm}}
=\pm\frac{\braket{\phi|\psi}}{\left|\braket{\psi|\phi}\right|}\sqrt{\frac{1\pm\left|\braket{\psi|\phi}\right|}{2}},
\end{align}
it holds that
\begin{align}
V^{\dagger}\ket{\psi}
=\sqrt{\frac{1+\left|\braket{\psi|\phi}\right|}{2}}\ket{0}+\sqrt{\frac{1-\left|\braket{\psi|\phi}\right|}{2}}\ket{1}
\end{align}
and
\begin{align}
V^{\dagger}\ket{\phi}
=\frac{\braket{\psi|\phi}}{\left|\braket{\psi|\phi}\right|}\left(\sqrt{\frac{1+\left|\braket{\psi|\phi}\right|}{2}}\ket{0}-\sqrt{\frac{1-\left|\braket{\psi|\phi}\right|}{2}}\ket{1}\right).
\end{align}
By defining $\ket{\psi'}=V^{\dagger}\ket{\psi}$ and
$\ket{\phi'}=V^{\dagger}\ket{\phi}$, we obtain
$\ket{\psi}=V\ket{\psi'}$ and $\ket{\phi}=V\ket{\phi'}$.
Moreover, we define the unitary operator
\begin{align}
W=\ket{0}\bra{x_{+}}+\ket{1}\bra{x_{-}},
\end{align}
where $\ket{x_{\pm}}$ is defined by
\begin{align}
\ket{x_{\pm}}=\frac{1}{\sqrt{2}}\left(\ket{x_{0}}\pm\ket{x_{1}}\right).
\end{align}
It is evident that
\begin{align}
W\ket{x_{0}}
=\ket{+}
\end{align}
and
\begin{align}
W\ket{x_{1}}
=\ket{-}
\end{align}
hold, where $\ket{\pm}=1/\sqrt{2}(\ket{0}\pm\ket{1})$.
Hence, $\ket{x_{0}}=W^{\dagger}\ket{+}$ and $\ket{x_{1}}=W^{\dagger}\ket{-}$ follow.
By employing $V$ and $W$, $\Psi$ can be rewritten as
\begin{align}
\Psi(A)
=V\Lambda(WAW^{\dagger})V^{\dagger},
\end{align}
where $\Lambda$ is defined by
\begin{align}
\Lambda(A)
=\ket{\psi'}\bra{+}A\ket{+}\bra{\psi'}+\ket{\phi'}\bra{-}A\ket{-}\bra{\phi'}.
\end{align}
From this expression, it holds that
\begin{align}
J_{\Psi}
=(W^{T}\otimes V)J_{\Lambda}((W^{T})^{\dagger}\otimes V^{\dagger}).
\end{align}
To calculate $J_{\Lambda}$, we need to derive the matrix representation of $\Lambda$.
The Pauli matrices are mapped by $\Lambda$ as follows:
\begin{align}
\Lambda(\sigma_{0})&=\ket{\psi'}\bra{\psi'}+\ket{\phi'}\bra{\phi'}, \\
\Lambda(\sigma_{1})&=\ket{\psi'}\bra{\psi'}-\ket{\phi'}\bra{\phi'}, \\
\Lambda(\sigma_{2})&=\Lambda(\sigma_{3})=0.
\end{align}
From the definition of $\ket{\psi'}$ and $\ket{\phi'}$, we derive the equalities
\begin{align}
\ket{\psi'}\bra{\psi'}+\ket{\phi'}\bra{\phi'}
=I+\left|\braket{\psi|\phi}\right|\sigma_{3}
\end{align}
and
\begin{align}
\ket{\psi'}\bra{\psi'}-\ket{\phi'}\bra{\phi'}
=\sqrt{1-\left|\braket{\psi|\phi}\right|^{2}}\sigma_{1}.
\end{align}
Thus, the matrix representation of $\Lambda$ is given by
\begin{equation}
\begin{pmatrix}
1 & 0 & 0 & 0 \\
0 & \sqrt{1-\left|\braket{\psi|\phi}\right|^{2}} & 0 & 0 \\
0 & 0 & 0 & 0 \\
\left|\braket{\psi|\phi}\right| & 0 & 0 & 0 \\
\end{pmatrix}
.
\end{equation}
From (\ref{eq:P'_dec}), it holds that
\begin{align}
J_{\Lambda}
=\frac{1}{2}\left(\sigma_{0}\otimes\sigma_{0}+\left|\braket{\psi|\phi}\right|\sigma_{0}\otimes\sigma_{3}+\sqrt{1-\left|\braket{\psi|\phi}\right|^{2}}\sigma_{1}\otimes\sigma_{1}\right).
\end{align}
Therefore, $J_{\Psi}$ is computed as
\begin{align}
J_{\Psi}
&=\frac{1}{2}\Bigl(\sigma_{0}\otimes\sigma_{0}+\left|\braket{\psi|\phi}\right|\sum^{3}_{j=1}O(V)_{j3}\sigma_{0}\otimes\sigma_{j} \\
&+\sqrt{1-\left|\braket{\psi|\phi}\right|^{2}}\sum^{3}_{i,j=1}O(W^{T})_{i1}O(V)_{j1}\sigma_{i}\otimes \sigma_{j}\Bigr),
\end{align}
where $O(W^{T})$ and $O(V)$ are defined by $O(W^{T})_{ik}=(1/2)\tr[\sigma_{i}W^{T}\sigma_{k}(W^{T})^{\dagger}]$ and $O(V)_{jl}=(1/2)\tr[\sigma_{j}W^{T}\sigma_{l}(W^{T})^{\dagger}]$, respectively.
It is evident that $O(W^{T})$ and $O(V)$ are orthonormal matrices.
Using this expression, we can calculate the function to be maximized as
\begin{align}
\tr[\tau J_{\Psi}]
=\frac{1}{2}\Bigl(1+\sin{u}\left<\bm{o}(V)_{3},\bm{b}(\tau_{2})\right>
+\cos{u}\left<\bm{o}(W^{T})_{1},N(\tau)\bm{o}(V)_{1}\right>\Bigr),
\end{align}
where $\mathbf{o}(\cdot)_{i}$ denotes the $i$th column vector of $O(\cdot)$ and $\sin{u}$ is defined by $\sin{u}=\left|\braket{\psi|\phi}\right|$ with $0\leq u\leq\pi/2$.
Consequently, $\tr[\tau J_{\Psi}]$ can be estimated from above as
\begin{align}
\tr[\tau J_{\Psi}]
&=\frac{1}{2}\Bigl(1+\sin{u}\left<\bm{o}(V)_{3},\bm{b}(\tau_{2})\right>
+\cos{u}\left<\bm{o}(W^{T})_{1},N(\tau)\bm{o}(V)_{1}\right>\Bigr) \\
&\leq\frac{1}{2}\Bigl(1+\sin{u}\left|\left<\bm{o}(V)_{3},\bm{b}(\tau_{2})\right>\right| \label{eq:gebc_est_mid}
+\cos{u}\left|\left<\bm{o}(W^{T})_{1},N(\tau)\bm{o}(V)_{1}\right>\right|\Bigr) \\
&\leq\frac{1}{2}\Bigl(1+\sin{u}\left\|\bm{b}(\tau_{2})\right\|
+\cos{u}\left\|N(\tau)\right\|\Bigr) \\
&\leq\frac{1}{2}\Bigl(1+\sqrt{\left\|\bm{b}(\tau_{2})\right\|^{2}
+\left\|N(\tau)\right\|^{2}}\Bigr)
\end{align}
where we applied the Cauchy–Schwarz inequality to estimate the second term of (\ref{eq:gebc_est_mid}) and we utilized the fact that the Euclidean norm of the column vectors of an orthonormal matrix is unity.
The estimation of the third term of (\ref{eq:gebc_est_mid}) employs a method that is similar to the one used in (\ref{eq:sec_term_last_est}).
\end{proof}

\subsection{General quantum channels}
In this subsection, we describe the treatment of general quantum channels.
First, we designate the set of quantum channels as $\mathcal{C}$.
In Ref. \cite{BETHRUSKAI2002159}, it is demonstrated that an arbitrary quantum channel can be represented by a convex combination of two elements in the so-called generalized extreme points.
Here, an arbitrary quantum channel in the generalized extreme points is expressed as
\begin{align}
\Psi=\Psi_{V}\circ\Lambda\circ\Psi_{W}, \label{eq:gep_dec}
\end{align}
where $\Psi_{V}$ and $\Psi_{W}$ represent the unitary channels implemented by unitary operators $V$ and $W$, respectively.
The matrix representation of $\Lambda$ is
\begin{equation}
T=
\begin{pmatrix}
1 & 0 & 0 & 0\\
0 & \cos{u_{1}} & 0 & 0\\
0 & 0 & \cos{u_{2}} & 0\\
\sin{u_{1}}\sin{u_{2}} & 0 & 0 & \cos{u_{1}}\cos{u_{2}}\\
\end{pmatrix}
, \label{eq:gep_mat}
\end{equation}
with $u_{1}\in[0,2\pi)$ and $u_{2}\in[0,\pi)$.
Similar to entanglement breaking channels, it is adequate to consider the generalized extreme points for evaluating $C(\tau,\mathcal{C})$.
Consequently, the following proposition is presented.
\begin{pro} \label{pro:gc}
In the case of $\mathcal{C}$, (\ref{eq:def_ub}) is bounded from above as
\begin{equation}
C(\tau,\mathcal{C})\leq\frac{1}{2}\left(1+\sqrt{\|\bm{b}(\tau_{2})\|^{2}+\|N(\tau)\|^{2}_{KF}}\right), \label{eq:bound_ch}
\end{equation}
where $\bm{b}(\cdot)$ is the Bloch vector defined in (\ref{eq:def_bloch}), and $\|\bm{b}(\cdot)\|$ is the Euclidean norm defined in (\ref{eq:def_euc_norm}).
The matrix $N(\cdot)$ is the correlation matrix defined in (\ref{eq:def_corr_mat}), and $\|N(\cdot)\|_{KF}$ is the Ky Fan norm defined in (\ref{eq:def_kf_norm}).
\end{pro}
\begin{proof}
Let $\Psi$ be a quantum channel in the generalized extreme points expressed by (\ref{eq:gep_dec}) and (\ref{eq:gep_mat}).
Similar to entanglement breaking channels, it holds that
\begin{align}
J_{\Psi}=(W^{T}\otimes V)J_{\Lambda}((W^{T})^{\dagger}\otimes V^{\dagger}).
\end{align}
Here, from (\ref{eq:gep_mat}), $J_{\Lambda}$ is expressed as
\begin{align}
J_{\Lambda}
&=\frac{1}{2}(\sigma_{0}\otimes\sigma_{0}+\sin{u_{1}}\sin{u_{2}}\sigma_{0}\otimes\sigma_{3}+\cos{u_{1}}\sigma_{1}\otimes\sigma_{1} \nonumber \\
&\quad-\cos{u_{2}}\sigma_{2}\otimes\sigma_{2}+\cos{u_{1}}\cos{u_{2}}\sigma_{3}\otimes\sigma_{3}).
\end{align}
Thus, $J_{\Psi}$ can be rewritten as
\begin{align}
J_{\Psi}
&=\frac{1}{2}(\sigma_{0}\otimes\sigma_{0}+\sin{u_{1}}\sin{u_{2}}\sum^{3}_{j=1}O(V)_{j3}\sigma_{0}\otimes\sigma_{j} \nonumber \\
&\quad+\cos{u_{1}}\sum^{3}_{i,j=1}O(W^{T})_{i1}O(V)_{j1}\sigma_{i}\otimes\sigma_{j} \nonumber \\
&\quad-\cos{u_{2}}\sum^{3}_{i,j=1}O(W^{T})_{i2}O(V)_{j2}\sigma_{i}\otimes\sigma_{j} \nonumber \\
&\quad+\cos{u_{1}}\cos{u_{2}}\sum^{3}_{i,j=1}O(W^{T})_{i3}O(V)_{j3}\sigma_{i}\otimes\sigma_{j},
\end{align}
where $O(W^{T})_{ik}$ and $O(V)_{jl}$ are defined by $O(W^{T})_{ik}=(1/2)\tr[\sigma_{i}W^{T}\sigma_{k}(W^{T})^{\dagger}]$ and $O(V)_{jl}=(1/2)\tr[\sigma_{j}V\sigma_{l}V^{\dagger}]$, respectively.
It is evident that $O(W^{T})_{ik}$ and $O(V)_{jl}$ are orthonormal matrices.
From this expression, we can obtain the following equality:
\begin{align}
\tr\left[\tau J_{\Psi}\right]
=\frac{1}{2}\left(1+\sin{u_{1}}\sin{u_{2}}\left<\bm{o}(V)_{3},\bm{b}(\tau_{2})\right>+\tr[AN(\tau)]\right),
\end{align}
where $\bm{o}(\cdot)_{i}$ is the $i$th column vectors of $O(\cdot)$ and $A$ is defined by
\begin{align}
A
&=\cos{u_{1}}\bm{o}(V)_{1}\bm{o}(W^{T})^{T}_{1} \nonumber \\
&-\cos{u_{2}}\bm{o}(V)_{2}\bm{o}(W^{T})^{T}_{2} \nonumber \\
&+\cos{u_{1}}\cos{u_{2}}\bm{o}(V)_{3}\bm{o}(W^{T})^{T}_{3}.
\end{align}
As the column vectors of an orthonormal matrix form an orthonormal base, it is obvious that
\begin{align}
&AA^{T} \nonumber \\
&=\cos^{2}{u_{1}}\bm{o}(V)_{1}\bm{o}(V)^{T}_{1}
+\cos^{2}{u_{2}}\bm{o}(V)_{2}\bm{o}(V)^{T}_{2}
+\cos^{2}{u_{1}}\cos^{2}{u_{2}}\bm{o}(V)_{3}\bm{o}(V)^{T}_{3}.
\end{align}
Hence, the set of singular values of $A$ is $\{\left|\cos{u_{1}}\right|,\left|\cos{u_{2}}\right|,\left|\cos{u_{1}}\cos{u_{2}}\right|\}$.
Consequently, the function to be maximized is bounded from above as
\begin{align}
\tr[\tau J_{\Psi}]
&=\frac{1}{2}\left(1+\sin{u_{1}}\sin{u_{2}}\left<\bm{o}(V)_{3},\bm{b}(\tau_{2})\right>+\tr[AN(\tau)]\right) \\
&\leq\frac{1}{2}\left(1+\left|\sin{u_{1}}\sin{u_{2}}\right|\left|\left<\bm{o}(V)_{3},\bm{b}(\tau_{2})\right>\right|+\left|\tr[AN(\tau)]\right|\right) \\
&\leq\frac{1}{2}\left(1+\left|\sin{u_{1}}\sin{u_{2}}\right|\|\bm{b}(\tau_{2})\|+\left|\tr[AN(\tau)]\right|\right) \\
&\leq\frac{1}{2}\left(1+\left|\sin{u_{1}}\sin{u_{2}}\right|\|\bm{b}(\tau_{2})\|+\sum^{3}_{i=1}s(A)_{i}s(N(\tau))_{i}\right),
\end{align}
where we employed the Cauchy–Schwarz inequality to derive the second inequality and the $s(\cdot)_{i}$ in the last expression represents the $i$th largest singular value of its argument.
The third inequality is a consequence of von Neumann's trace inequality whose elementary proof is available in Ref. \citen{Mirsky1975}.
For further estimation, let $\left|\cos{u_{j}}\right|=\max\{\left|\cos{u_{1}}\right|,\left|\cos{u_{2}}\right|\}$.
By employing this value, the following chain of inequalities holds:
\begin{align}
\tr[\tau J_{\Psi}]
&\leq\frac{1}{2}\left(1+\left|\sin{u_{1}}\sin{u_{2}}\right|\|\bm{b}(\tau_{2})\|+\left|\cos{u_{j}}\right|\|N(\tau)\|_{KF}\right) \\
&\leq\frac{1}{2}\left(1+\left|\sin{u_{j}}\right|\|\bm{b}(\tau_{2})\|+\left|\cos{u_{j}}\right|\|N(\tau)\|_{KF}\right) \\
&\leq\frac{1}{2}\left(1+\sqrt{\|\bm{b}(\tau_{2})\|^{2}+\|N(\tau)\|_{KF}^{2}}\right),
\end{align}
where we utilized the definition of the Ky Fan norm in the first inequality.
Moreover, an arbitrary quantum channel is also bounded from above by the last expression, as it can be expressed as a convex combination of two elements in the generalized extreme points.
\end{proof}

\subsection{Comparison of the upper bounds} \label{sec:comp_ub}
Let us compare the obtained upper bounds.
For simplicity, we denote them as follows:
\begin{align}
\Tilde{C}(\tau,\mathcal{D})&=C(\tau,\mathcal{D})=\frac{1}{2}\left(1+\|\bm{b}(\tau_{2})\|\right), \\
\Tilde{C}(\tau,\mathcal{R})
&=C(\tau,\mathcal{R})
=\frac{1}{2}\left(1+\left\|N(\tau)\right\|_{KF}\right), \\
\Tilde{C}(\tau,\mathcal{UE})&=\frac{1}{2}\left(1+\|N(\tau)\|\right), \\
\Tilde{C}(\tau,\mathcal{GE})&=\frac{1}{2}\left(1+\sqrt{\|\bm{b}(\tau_{2})\|^{2}+\|N(\tau)\|^{2}}\right), \\
\Tilde{C}(\tau,\mathcal{C})&=\frac{1}{2}\left(1+\sqrt{\|\bm{b}(\tau_{2})\|^{2}+\|N(\tau)\|^{2}_{KF}}\right).
\end{align}
It is evident that the relationships between upper bounds are consistent with the subset inclusions.
Namely, the inequalities
\begin{align}
\Tilde{C}(\tau,\mathcal{UE})\leq\Tilde{C}(\tau,\mathcal{R})\leq\Tilde{C}(\tau,\mathcal{C}), \\
\Tilde{C}(\tau,\mathcal{UE})\leq\Tilde{C}(\tau,\mathcal{E})\leq\Tilde{C}(\tau,\mathcal{C}), \\
\Tilde{C}(\tau,\mathcal{D})\leq\Tilde{C}(\tau,\mathcal{E})\leq\Tilde{C}(\tau,\mathcal{C})
\end{align}
reflect the inclusions
\begin{align}
&\mathcal{UE}\subseteq\mathcal{R}\subseteq\mathcal{C}, \\
&\mathcal{UE}\subseteq\mathcal{E}\subseteq\mathcal{C}, \\
&\mathcal{D}\subseteq\mathcal{E}\subseteq\mathcal{C},
\end{align}
respectively.
Therefore, we can consider $\Tilde{C}(\tau,\mathcal{C})$ as a universal upper bound and the others serve as specialized upper bounds.
In general, $\Tilde{C}(\tau,\mathcal{C})$ is an estimate of $C(\tau,\mathcal{C})$; however, in a specific case, it provides the exact value of $C(\tau,\mathcal{C})$.
\begin{pro} \label{pro:bound_c_sp}
If $\tau_{2}=I/2$, the equality
\begin{equation}
C(\tau,\mathcal{C})
=\frac{1}{2}\left(1+\left\|N(\tau)\right\|_{KF}\right) \label{eq:bound_c_sp}
\end{equation}
holds.
The matrix $N(\cdot)$ represents the correlation matrix defined in (\ref{eq:def_corr_mat}), and $\|N(\cdot)\|_{KF}$ is the Ky Fan norm defind in (\ref{eq:def_kf_norm}).
\end{pro}
\begin{proof}
From the definitions, it holds that
\begin{equation}
C(\tau,\mathcal{R})
\leq C(\tau,\mathcal{C})
\leq \Tilde{C}(\tau,\mathcal{C}). \label{eq:R_C_chain}
\end{equation}
If $\tau_{2}=I/2$, it is evident that
\begin{equation}
\|\bm{b}(\tau_{2})\|
=\|\bm{b}(I/2)\|
=0.
\end{equation}
Thus, (\ref{eq:R_C_chain}) becomes
\begin{equation}
\frac{1}{2}\left(1+\left\|N(\tau)\right\|_{KF}\right)
\leq C(\tau,\mathcal{C})
\leq\frac{1}{2}\left(1+\left\|N(\tau)\right\|_{KF}\right),
\end{equation}
which indicates that (\ref{eq:bound_c_sp}) holds.
\end{proof}
For example, $C(\tau,\mathcal{C})$ is determined by (\ref{eq:bound_c_sp}), if $\tau$ is a quantum state
\begin{equation}
\sum_{x}p_{x}\ket{\Psi_{x}}\bra{\Psi_{x}},
\end{equation}
where $\{p_{x}\}_{x}$ represents a probability distribution and all $\ket{\Psi_{x}}$ are maximally entangled states.

Any $\Tilde{C}(\tau,\mathcal{X})$ is defined by $\|\bm{b}(\tau_{2})\|$, $\left\|N(\tau)\right\|$, and $\left\|N(\tau)\right\|_{KF}$, thus making the calculation of $\Tilde{C}(\tau,\mathcal{X})$ analytically more feasible than $C(\tau,\mathcal{X})$.

\section{Joint convertibility} \label{sec:jc}
We can utilize the obtained upper bounds to quantify a notion of convertibility.
Let us consider two pairs of pure states $(\ket{\psi},\ket{\phi})$ and $(\ket{e},\ket{f})$, and a specific class of quantum channels $\mathcal{X}$.
We say that $(\ket{\psi},\ket{\phi})$ is jointly convertible into $(\ket{e},\ket{f})$ within $\mathcal{X}$ if there exists a quantum channel $\Psi\in\mathcal{X}$ such that
\begin{equation}
\ket{e}\bra{e}=\Psi(\ket{\psi}\bra{\psi}) \label{eq:jc_cond_1}
\end{equation}
and
\begin{equation}
\ket{f}\bra{f}=\Psi(\ket{\phi}\bra{\phi}). \label{eq:jc_cond_2}
\end{equation}
If $\left|\braket{e|f}\right|<\left|\braket{\psi|\phi}\right|$, there is no quantum channel that satisfies (\ref{eq:jc_cond_1}) and (\ref{eq:jc_cond_2}).
This fact can be readily demonstrated by employing the fidelity, which is defined for arbitrary density operators $\rho$ and $\sigma$ as
\begin{equation}
F(\rho,\sigma)=\tr\sqrt{\sqrt{\sigma}\rho\sqrt{\sigma}}.
\end{equation}
The fidelity has the property
\begin{equation}
F(\ket{\psi}\bra{\psi},\ket{\phi}\bra{\phi})
\leq F(\Psi(\ket{\psi}\bra{\psi}),\Psi(\ket{\phi}\bra{\phi})),
\end{equation}
where $\Psi$ denotes an arbitrary quantum channel \cite{nielsen_chuang_2010}.
Hence, if there exists $\Psi_{0}$ that satisfies (\ref{eq:jc_cond_1}) and (\ref{eq:jc_cond_2}), the property of the fidelity implies that
\begin{equation}
\left|\braket{\psi|\phi}\right|
\leq F(\Psi_{0}(\ket{\psi}\bra{\psi}),\Psi_{0}(\ket{\phi}\bra{\phi}))
=\left|\braket{e|f}\right|,
\end{equation}
which contradicts $\left|\braket{e|f}\right|<\left|\braket{\psi|\phi}\right|$.
Furthermore, we present the following proposition.
\begin{pro}
If $\left|\braket{\psi|\phi}\right|\leq\left|\braket{e|f}\right|$, then there exists a quantum channel $\Psi_{0}$ that satisfies (\ref{eq:jc_cond_1}) and (\ref{eq:jc_cond_2}).
\end{pro}
Indeed, if $\left|\braket{e|f}\right|=1$, the depolarizing channel that outputs $\ket{e}\bra{e}$ serves as this quantum channel.
If $\left|\braket{\psi|\phi}\right|=0$, it is implemented by the entanglement breaking channel defined as
\begin{equation}
\ket{e}\braket{\psi|A|\psi}\bra{e}+\ket{f}\braket{\phi|A|\phi}\bra{f}
\end{equation}
for an arbitrary linear operator $A$.
In the remaining case where $0<\left|\braket{\psi|\phi}\right|\leq\left|\braket{e|f}\right|<1$, it is provided by the quantum channel in the generalized extreme points, as defined by
\begin{equation}
\Psi_{V}\circ\Lambda\circ\Psi_{U}.
\end{equation}
Here, $V$ and $U$ are defined as
\begin{align}
V&=\ket{e_{+}}\bra{0}+\ket{e_{-}}\bra{1}, \\
U&=\ket{0}\bra{\psi_{+}}+\ket{1}\bra{\psi_{-}},
\end{align}
where the pure states are defined as
\begin{align}
\ket{e_{\pm}}&=\frac{1}{\sqrt{2\left(1\pm\left|\braket{e|f}\right|\right)}}\left(\ket{e}\pm\frac{\braket{f|e}}{\left|\braket{e|f}\right|}\ket{f}\right), \\
\ket{\psi_{\pm}}&=\frac{1}{\sqrt{2\left(1\pm\left|\braket{\psi|\phi}\right|\right)}}\left(\ket{\psi}\pm\frac{\braket{\phi|\psi}}{\left|\braket{\psi|\phi}\right|}\ket{\phi}\right).
\end{align}
The parameters of $\Lambda$ in (\ref{eq:gep_mat}) are defined as
\begin{align}
\cos{u_{1}}&=\sqrt{\frac{1-\left|\braket{e|f}\right|^{2}}{1-\left|\braket{\psi|\phi}\right|^{2}}}, \\
\cos{u_{2}}&=\frac{\left|\braket{\psi|\phi}\right|}{\left|\braket{e|f}\right|}\sqrt{\frac{1-\left|\braket{e|f}\right|^{2}}{1-\left|\braket{\psi|\phi}\right|^{2}}}, \\
\sin{u_{1}}\sin{u_{2}}&=\frac{\left|\braket{e|f}\right|^{2}-\left|\braket{\psi|\phi}\right|^{2}}{\left(1-\left|\braket{\psi|\phi}\right|^{2}\right)\left|\braket{e|f}\right|}.
\end{align}
Therefore, we can rephrase the joint convertibility within $\mathcal{C}$ using the following proposition, which represents a special case of the theorem formulated in Ref. \citen{ALBERTI1980163}.
\begin{pro} \label{pro:jc_c_iff}
$(\ket{\psi},\ket{\phi})$ is jointly convertible into $(\ket{e},\ket{f})$ within the set of quantum channels $\mathcal{C}$ if and only if $\left|\braket{\psi|\phi}\right|\leq\left|\braket{e|f}\right|$.
\end{pro}

To quantify the joint convertibility within a subset $\mathcal{X}\subseteq\mathcal{C}$, we can employ the derived upper bounds.
Let us consider two PPOVMs
\begin{align}
\bm{S}
&=
\left\{\ket{\psi}\bra{\psi}^{T}\otimes\ket{e}\bra{e},\ket{\psi}\bra{\psi}^{T}\otimes\left(I-\ket{e}\bra{e}\right)\right\} \\
&=\left\{S_{m_{0}},S_{m_{1}}\right\}
\end{align}
and
\begin{align}
\bm{T}
&=
\left\{\ket{\phi}\bra{\phi}^{T}\otimes\ket{f}\bra{f},\ket{\phi}\bra{\phi}^{T}\otimes\left(I-\ket{f}\bra{f}\right)\right\} \\
&=\left\{T_{n_{0}},T_{n_{1}}\right\}.
\end{align}
The PPOVM $\bm{S}$ corresponds to the ancilla-free measurement that uses $\ket{\psi}$ as the input state and $\left\{\ket{e}\bra{e}, I-\ket{e}\bra{e}\right\}$ as the POVM measurement.
Similarly, $\bm{T}$ is the ancilla-free measurement that uses $\ket{\phi}$ as the input state and $\left\{\ket{f}\bra{f}, I-\ket{f}\bra{f}\right\}$ as the POVM measurement.
Let us assume that a quantum channel $\Psi\in\mathcal{X}$ is measured randomly by $\bm{S}$ and $\bm{T}$.
The probability distributions associated with $\bm{S}$ and $\bm{T}$ are represented by $\left\{p_{m_{0}}(\Psi),p_{m_{1}}(\Psi)\right\}$ and $\left\{q_{n_{0}}(\Psi),q_{n_{1}}(\Psi)\right\}$, respectively.
The probability of obtaining $m_{0}$ or $n_{0}$ is given by
\begin{align}
\frac{1}{2}p_{m_{0}}(\Psi)+\frac{1}{2}q_{n_{0}}(\Psi)
&=\frac{1}{2}\braket{e|\Psi(\ket{\psi}\bra{\psi})|e}+\frac{1}{2}\braket{f|\Psi(\ket{\phi}\bra{\phi})|f} \\
&=\frac{1}{2}F(\ket{e}\bra{e},\Psi(\ket{\psi}\bra{\psi}))^{2}+\frac{1}{2}F(\ket{f}\bra{f},\Psi(\ket{\phi}\bra{\phi}))^{2}.
\end{align}
Evidently, the last equality indicates that this quantity represents the average (squared) fidelity, meaning that the average closeness between the inputs and outputs.
Thus, its maximum
\begin{equation}
\max_{\Psi\in\mathcal{X}}\left\{\frac{1}{2}p_{m_{0}}(\Psi)+\frac{1}{2}q_{n_{0}}(\Psi)\right\} \label{eq:m0_n0_max}
\end{equation}
indicates, within $\mathcal{X}$, how close $\Psi(\ket{\psi}\bra{\psi})$ and $\Psi(\ket{\phi}\bra{\phi})$ can be to $\ket{e}\bra{e}$ and $\ket{f}\bra{f}$ at the same time.
Indeed, it is simple to verify that $(\ket{\psi},\ket{\phi})$ is jointly convertible into $(\ket{e},\ket{f})$ within $\mathcal{X}$ if and only if (\ref{eq:m0_n0_max}) equals unity.
Therefore, the problem of the condition for the joint convertibility is translated into the problem of estimating (\ref{eq:m0_n0_max}).
We can rewrite (\ref{eq:m0_n0_max}) as
\begin{align}
&\max_{\Psi\in\mathcal{X}}\left\{\frac{1}{2}p_{m_{0}}(\Psi)+\frac{1}{2}q_{n_{0}}(\Psi)\right\} \\
&=\max_{\Psi\in\mathcal{X}}\left\{\tr\left[\frac{1}{2}\left(\ket{\psi}\bra{\psi}^{T}\otimes\ket{e}\bra{e}+\ket{\phi}\bra{\phi}^{T}\otimes\ket{f}\bra{f}\right)J_{\Psi}\right]\right\} \\
&=C(\tau,\mathcal{X}),
\end{align}
where we defined the density operator $\tau$ as
\begin{equation}
\tau=\frac{1}{2}\left(\ket{\psi}\bra{\psi}^{T}\otimes\ket{e}\bra{e}+\ket{\phi}\bra{\phi}^{T}\otimes\ket{f}\bra{f}\right).
\end{equation}
Consequently, the upper bounds (\ref{eq:bound_gd}), (\ref{eq:bound_unital}), (\ref{eq:bound_ue}), (\ref{eq:bound_ge}), and (\ref{eq:bound_ch}) define the conditions of the joint convertibility for each subset, and can be considered measures of the joint convertibility.
Note that these bounds are expressed in terms of the singular values of $N(\tau)$ and $\|\bm{b}(\tau_{2})\|$, where $\tau_{2}=\tr_{1}\tau$.
Thus, we only need to calculate these values.
To obtain the singular values of $N(\tau)$, we use the following lemma, with the proof provided in Appendix.
\begin{lem} \label{lem:sing_inv}
Let $\rho$ be an arbitrary quantum state in a two-qubit system.
The singular values of $N(\rho)$ are unchanged under transformations by local unitary channels.
In other words, the set of the singular values of $N(\rho)$ is equal to that of $N((V_{1}\otimes V_{2})\rho(V^{\dagger}_{1}\otimes V^{\dagger}_{2}))$, where $V_{1}$ and $V_{2}$ are unitary operators.
\end{lem}
We can now efficiently calculate the singular values of $N(\tau)$ using this lemma.
\begin{pro}
The set of the singular values of $N(\tau)$ is
\begin{align}
\left\{\left|\braket{\psi|\phi}\right|\left|\braket{e|f}\right|,\sqrt{\left(1-\left|\braket{\psi|\phi}\right|^{2}\right)\left(1-\left|\braket{e|f}\right|^{2}\right)},0\right\}. \label{eq:sv_corr}
\end{align}
\end{pro}
\begin{proof}
Based on Lemma \ref{lem:sing_inv}, we can use local unitary channels to $\tau$ to simplify the calculation of the singular values of $N(\tau)$.
Let $V_{1}$ and $V_{2}$ be arbitrary unitary operators.
By defining
\begin{align}
\tau'
=((V^{\dagger}_{1})^{T}\otimes V_{2})\tau(V^{T}_{1}\otimes V^{\dagger}_{2}), \label{eq:def_tau'}
\end{align}
it holds that
\begin{align}
\tau'
&=\frac{1}{2}\left(\left(\ket{\psi'}\bra{\psi'}\right)^{T}\otimes\ket{e'}\bra{e'}+\left(\ket{\phi'}\bra{\phi'}\right)^{T}\otimes\ket{f'}\bra{f'}\right),
\end{align}
where $\ket{\psi'}=V_{1}\ket{\psi}$, $\ket{\phi'}=V_{1}\ket{\phi}$, $\ket{e'}=V_{2}\ket{e}$, and $\ket{f'}=V_{2}\ket{f}$.
Additionally, for an arbitrary pure state $\ket{x}$ in a qubit system, by defining its complex conjugate as
\begin{align}
\ket{\overline{x}}=\braket{x|0}\ket{0}+\braket{x|1}\ket{1},
\end{align}
we obtain
\begin{align}
\tau'
=\frac{1}{2}\left(\ket{\overline{\psi'}}\bra{\overline{\psi'}}\otimes\ket{e'}\bra{e'}+\ket{\overline{\phi'}}\bra{\overline{\phi'}}\otimes\ket{f'}\bra{f'}\right).
\end{align}
From the definition of $N(\cdot)$, it holds that
\begin{align}
N(\tau')_{ij}
&=\frac{1}{2}\left(b(\overline{\psi'})_{i}b(e')_{j}+b(\overline{\phi'})_{i}b(f')_{j}\right), \label{eq:tra_corr}
\end{align}
where $b(\cdot)_{i}$ represents the $i$th element of the Bloch vector defined by $b(x)_{i}=\braket{x|\sigma_{i}|x}$ for an arbitrary state $\ket{x}$.
Let the unitary operators be
\begin{align}
V_{1}&=\ket{0}\bra{\psi_{+}}+\ket{1}\bra{\psi_{-}}, \label{eq:def_v1} \\
V_{2}&=\ket{0}\bra{e_{+}}+\ket{1}\bra{e_{-}}, \label{eq:def_v2}
\end{align}
where the pure states are defined as
\begin{align}
\ket{\psi_{\pm}}&=\frac{1}{\sqrt{2(1\pm\left|\braket{\psi|\phi}\right|)}}\left(\ket{\psi}\pm\frac{\braket{\phi|\psi}}{\left|\braket{\psi|\phi}\right|}\ket{\phi}\right), \label{eq:psi_pm} \\
\ket{e_{\pm}}&=\frac{1}{\sqrt{2(1\pm\left|\braket{e|f}\right|)}}\left(\ket{e}\pm\frac{\braket{f|e}}{\left|\braket{e|f}\right|}\ket{f}\right). \label{eq:e_pm}
\end{align}
We assumed that $0<\left|\braket{\psi|\phi}\right|<1$ and $0<\left|\braket{e|f}\right|<1$ because the calculation is more simple when one of the conditions does not hold.
It is evident that
\begin{align}
\braket{\psi_{\pm}|\psi}
&=\sqrt{\frac{1\pm\left|\braket{\psi|\phi}\right|}{2}}, \\
\braket{\psi_{\pm}|\phi}
&=\pm\frac{\braket{\psi|\phi}}{\left|\braket{\psi|\phi}\right|}\sqrt{\frac{1\pm\left|\braket{\psi|\phi}\right|}{2}},
\end{align}
and by a similar calculation, it is also valid that
\begin{align}
\braket{e_{\pm}|e}&=\sqrt{\frac{1\pm\left|\braket{e|f}\right|}{2}}, \\
\braket{e_{\pm}|f}&=\pm\frac{\braket{e|f}}{\left|\braket{e|f}\right|}\sqrt{\frac{1\pm\left|\braket{e|f}\right|}{2}}.
\end{align}
Hence, $V_{1}$ acts $\ket{\psi}$ and $\ket{\phi}$ as
\begin{align}
V_{1}\ket{\psi}
&=\sqrt{\frac{1+\left|\braket{\psi|\phi}\right|}{2}}\ket{0}+\sqrt{\frac{1-\left|\braket{\psi|\phi}\right|}{2}}\ket{1}, \\
V_{1}\ket{\phi}
&=\frac{\braket{\psi|\phi}}{\left|\braket{\psi|\phi}\right|}\left(\sqrt{\frac{1+\left|\braket{\psi|\phi}\right|}{2}}\ket{0}-\sqrt{\frac{1-\left|\braket{\psi|\phi}\right|}{2}}\ket{1}\right),
\end{align}
and $V_{2}$ also acts $\ket{e}$ and $\ket{f}$ as
\begin{align}
V_{2}\ket{e}
&=\sqrt{\frac{1+\left|\braket{e|f}\right|}{2}}\ket{0}+\sqrt{\frac{1-\left|\braket{e|f}\right|}{2}}\ket{1}, \\
V_{2}\ket{f}
&=\frac{\braket{e|f}}{\left|\braket{e|f}\right|}\left(\sqrt{\frac{1+\left|\braket{e|f}\right|}{2}}\ket{0}-\sqrt{\frac{1-\left|\braket{e|f}\right|}{2}}\ket{1}\right).
\end{align}
Evidently, it holds that $\ket{\overline{\psi'}}\bra{\overline{\psi'}}=\ket{\psi'}\bra{\psi'}$ and $\ket{\overline{\phi'}}\bra{\overline{\phi'}}=\ket{\phi'}\bra{\phi'}$.
Therefore, we can obtain the Bloch vectors as
\begin{align}
\bm{b}(\psi')&=(\sqrt{1-\left|\braket{\psi|\phi}\right|^{2}},0,\left|\braket{\psi|\phi}\right|), \\
\bm{b}(\phi')&=(-\sqrt{1-\left|\braket{\psi|\phi}\right|^{2}},0,\left|\braket{\psi|\phi}\right|), \\
\bm{b}(e')&=(\sqrt{1-\left|\braket{e|f}\right|^{2}},0,\left|\braket{e|f}\right|), \\
\bm{b}(f')&=(-\sqrt{1-\left|\braket{e|f}\right|^{2}},0,\left|\braket{e|f}\right|).
\end{align}
Hence, $N(\tau')$ is expressed as
\begin{equation}
N(\tau')=
\begin{pmatrix}
\sqrt{1-\left|\braket{\psi|\phi}\right|^{2}}\sqrt{1-\left|\braket{e|f}\right|^{2}} & 0 & 0 \\
0 & 0 & 0 \\
0 & 0 & \left|\braket{\psi|\phi}\right|\left|\braket{e|f}\right| \\
\end{pmatrix}
,
\end{equation}
and it is evident that the set of the singular values of $N(\tau')$ is
\begin{align}
\left\{\left|\braket{\psi|\phi}\right|\left|\braket{e|f}\right|,\sqrt{\left(1-\left|\braket{\psi|\phi}\right|^{2}\right)\left(1-\left|\braket{e|f}\right|^{2}\right)},0\right\}.
\end{align}
Lemma \ref{lem:sing_inv} guarantees that the set of the singular values of $N(\tau)$ remains unchanged under the application of the unitary channel given by $V^{T}_{1}\otimes V^{\dagger}_{2}$.
While we assumed that $0<\left|\braket{\psi|\phi}\right|<1$ and $0<\left|\braket{e|f}\right|<1$, the result is also applicable if one of the conditions does not hold.
In fact, the same calculation can be performed by replacing $\ket{\psi_{\pm}}$ in (\ref{eq:psi_pm}) with
\begin{equation}
\ket{\psi_{\pm}}=\frac{1}{\sqrt{2}}\left(\ket{\psi}\pm\ket{\phi}\right)
\end{equation}
if $\left|\braket{\psi|\phi}\right|=0$, and
\begin{align}
\ket{\psi_{+}}&=\ket{\psi} \\
\ket{\psi_{-}}&=\ket{\psi^{\perp}}
\end{align}
if $\left|\braket{\psi|\phi}\right|=1$.
Here, $\ket{\psi^{\perp}}$ is a unit vector orthogonal to $\ket{\psi}$.
Likewise, one simply needs to replace $\ket{e_{\pm}}$ in (\ref{eq:e_pm}) with
\begin{equation}
\ket{e_{\pm}}=\frac{1}{\sqrt{2}}\left(\ket{e}\pm\ket{f}\right)
\end{equation}
if $\left|\braket{e|f}\right|=0$, and
\begin{align}
\ket{e_{+}}&=\ket{e} \\
\ket{e_{-}}&=\ket{e^{\perp}}
\end{align}
if $\left|\braket{e|f}\right|=1$.
Here, $\ket{e^{\perp}}$ is a unit vector orthogonal to $\ket{e}$.
\end{proof}
It is evident that the norm of the Bloch vector of
\begin{align}
\tau_{2}
&=\tr_{1}\tau \\
&=\frac{1}{2}\left(\ket{e}\bra{e}+\ket{f}\bra{f}\right)
\end{align}
is given by
\begin{equation}
\left\|\bm{b}(\tau_{2})\right\|=\left|\braket{e|f}\right|. \label{eq:bloch_norm}
\end{equation}
Using (\ref{eq:sv_corr}) and (\ref{eq:bloch_norm}), we can examine the joint convertibility for the specific classes discussed in the previous section.

\subsection{Depolarizing channels}
Based on Proposition \ref{pro:gdc} and (\ref{eq:bloch_norm}), the following corollary is established.
\begin{cor}
It holds that
\begin{equation}
\max_{\Psi\in\mathcal{D}}\left\{\frac{1}{2}p_{m_{0}}(\Psi)+\frac{1}{2}q_{n_{0}}(\Psi)\right\}
=\frac{1}{2}\left(1+\left|\braket{e|f}\right|\right), \label{eq:gd_jc}
\end{equation}
The expression on the right-hand side of (\ref{eq:gd_jc}) becomes unity if and only if $\left|\braket{e|f}\right|=1$.
Therefore, $(\ket{\psi}, \ket{\phi})$ is jointly convertible into $(\ket{e},\ket{f})$ within $\mathcal{D}$ if and only if $\left|\braket{e|f}\right|=1$.
\end{cor}
It is easy to verify that the maximum in (\ref{eq:gd_jc}) is achieved by the depolarizing channel which produces
\begin{equation}
\ket{e_{+}}=\frac{1}{\sqrt{2(1+\left|\braket{e|f}\right|)}}\left(\ket{e}+\frac{\braket{f|e}}{\left|\braket{e|f}\right|}\ket{f}\right).
\end{equation}
if $\left|\braket{e|f}\right|\neq0$.
If $\left|\braket{e|f}\right|=0$, it is achieved by an arbitrary depolarizing channel.

\subsection{Unital channels}
Based on Proposition \ref{eq:pro_unital} and (\ref{eq:sv_corr}), we obtain the following corollary.
\begin{cor}
It holds that
\begin{align}
&\max_{\Psi\in\mathcal{R}}\left\{\frac{1}{2}p_{m_{0}}(\Psi)+\frac{1}{2}q_{n_{0}}(\Psi)\right\} \nonumber \\
&=\frac{1}{2}\left(1+\left|\braket{\psi|\phi}\right|\left|\braket{e|f}\right|+\sqrt{\left(1-\left|\braket{\psi|\phi}\right|^{2}\right)\left(1-\left|\braket{e|f}\right|^{2}\right)}\right). \label{eq:unital_jc}
\end{align}
The expression on the right-hand side becomes unity if and only if $\left|\braket{\psi|\phi}\right|=\left|\braket{e|f}\right|$.
Therefore, $(\ket{\psi}, \ket{\phi})$ is jointly convertible into $(\ket{e},\ket{f})$ within $\mathcal{R}$ if and only if $\left|\braket{\psi|\phi}\right|=\left|\braket{e|f}\right|$.
\end{cor}
\begin{proof}
The necessary and sufficient condition can be shown using the arithmetic-geometric mean inequality.
Indeed, the inequality
\begin{align}
&\left|\braket{\psi|\phi}\right|\left|\braket{e|f}\right|+\sqrt{\left(1-\left|\braket{\psi|\phi}\right|^{2}\right)\left(1-\left|\braket{e|f}\right|^{2}\right)} \nonumber \\
&\leq\frac{\left|\braket{\psi|\phi}\right|^{2}+\left|\braket{e|f}\right|^{2}}{2}+\frac{1-\left|\braket{\psi|\phi}\right|^{2}+1-\left|\braket{e|f}\right|^{2}}{2} \nonumber \\
&=1
\end{align}
becomes equality if and only if $\left|\braket{\psi|\phi}\right|^{2}=\left|\braket{e|f}\right|^{2}$ and $1-\left|\braket{\psi|\phi}\right|^{2}=1-\left|\braket{e|f}\right|^{2}$
hold.
In other words, $\left|\braket{\psi|\phi}\right|=\left|\braket{e|f}\right|$.
\end{proof}
The maximum of (\ref{eq:unital_jc}) is reached by the unitary channel characterized by the unitary operator
\begin{equation}
\ket{e_{+}}\bra{\psi_{+}}+\ket{e_{-}}\bra{\psi_{-}}. \label{eq:unital_max}
\end{equation}
Here $\ket{\psi_{+}}$ and $\ket{\psi_{-}}$ are defined as
\begin{equation}
\ket{\psi_{\pm}}=\frac{1}{\sqrt{2(1\pm\left|\braket{\psi|\phi}\right|)}}\left(\ket{\psi}\pm\frac{\braket{\phi|\psi}}{\left|\braket{\psi|\phi}\right|}\ket{\phi}\right) \label{eq:psi_pm_def}
\end{equation}
for $0<\left|\braket{\psi|\phi}\right|<1$, and
\begin{equation}
\ket{\psi_{\pm}}=\frac{1}{\sqrt{2}}\left(\ket{\psi}\pm\ket{\phi}\right)
\end{equation}
for $\left|\braket{\psi|\phi}\right|=0$, and for $\left|\braket{\psi|\phi}\right|=1$
\begin{align}
\ket{\psi_{+}}&=\ket{\psi}, \\
\ket{\psi_{-}}&=\ket{\psi^{\perp}},
\end{align}
where $\ket{\psi^{\perp}}$ is a unit vector orthogonal to $\ket{\psi}$.
Likewise, $\ket{e_{+}}$ and $\ket{e_{-}}$ are defined as
\begin{equation}
\ket{e_{\pm}}=\frac{1}{\sqrt{2(1\pm\left|\braket{e|f}\right|)}}\left(\ket{e}\pm\frac{\braket{e|f}}{\left|\braket{e|f}\right|}\ket{f}\right)
\end{equation}
for $0<\left|\braket{e|f}\right|<1$, and
\begin{equation}
\ket{e_{\pm}}=\frac{1}{\sqrt{2}}\left(\ket{e}\pm\ket{f}\right)
\end{equation}
for $\left|\braket{e|f}\right|=0$, and for $\left|\braket{e|f}\right|=1$
\begin{align}
\ket{e_{+}}&=\ket{e}, \\
\ket{e_{-}}&=\ket{e^{\perp}}, \label{eq:e_-_def}
\end{align}
where $\ket{e^{\perp}}$ is a unit vector orthogonal to $\ket{e}$.

\subsection{Unital entanglement breaking channels}
Although Proposition \ref{pro:ueb} is stated in terms of an inequality, we derive the following corollary.
\begin{cor}
It holds that
\begin{align}
&\max_{\Psi\in\mathcal{UE}}\left\{\frac{1}{2}p_{m_{0}}(\Psi)+\frac{1}{2}q_{n_{0}}(\Psi)\right\} \nonumber \\
&=\frac{1}{2}\left(1+\max\left\{\left|\braket{\psi|\phi}\right|\left|\braket{e|f}\right|,\sqrt{\left(1-\left|\braket{\psi|\phi}\right|^{2}\right)\left(1-\left|\braket{e|f}\right|^{2}\right)}\right\}\right). \label{eq:ueb_jc}
\end{align}
The expression on the right-hand side becomes unity if and only if $\left|\braket{\psi|\phi}\right|=\left|\braket{e|f}\right|=1$ or $\left|\braket{\psi|\phi}\right|=\left|\braket{e|f}\right|=0$.
Therefore, $(\ket{\psi}, \ket{\phi})$ is jointly convertible into $(\ket{e},\ket{f})$ within $\mathcal{UE}$ if and only if $\left|\braket{\psi|\phi}\right|=\left|\braket{e|f}\right|=1$ or $\left|\braket{\psi|\phi}\right|=\left|\braket{e|f}\right|=0$.
\end{cor}
Indeed, if $\left|\braket{\psi|\phi}\right|\left|\braket{e|f}\right|\geq\sqrt{\left(1-\left|\braket{\psi|\phi}\right|^{2}\right)\left(1-\left|\braket{e|f}\right|^{2}\right)}$, then the right-hand side of (\ref{eq:ueb_jc}) is achieved by the unital entanglement breaking channel defined as
\begin{equation}
\Psi_{0}(A)
=\ket{e_{+}}\braket{\psi_{+}|A|\psi_{+}}\bra{e_{+}}+\ket{e_{-}}\braket{\psi_{-}|A|\psi_{-}}\bra{e_{-}},
\end{equation}
where $A$ represents an arbitrary linear operator, and $\ket{\psi_{\pm}}$ and $\ket{e_{\pm}}$ are the pure states defined from (\ref{eq:psi_pm_def}) to (\ref{eq:e_-_def}).
If $\left|\braket{\psi|\phi}\right|\left|\braket{e|f}\right|<\sqrt{\left(1-\left|\braket{\psi|\phi}\right|^{2}\right)\left(1-\left|\braket{e|f}\right|^{2}\right)}$, it is achieved by
\begin{equation}
\Psi_{1}(A)
=\ket{e'_{+}}\braket{\psi'_{+}|A|\psi'_{+}}\bra{e'_{+}}+\ket{e'_{-}}\braket{\psi'_{-}|A|\psi'_{-}}\bra{e'_{-}}, \label{eq:ueb_max}
\end{equation}
where $A$ represents an arbitrary linear operator, and the pure states are defined as
\begin{align}
\ket{\psi'_{\pm}}&=\frac{1}{\sqrt{2}}\left(\ket{\psi_{+}}\pm\ket{\psi_{-}}\right), \\
\ket{e'_{\pm}}&=\frac{1}{\sqrt{2}}\left(\ket{e_{+}}\pm\ket{e_{-}}\right).
\end{align}
Furthermore, in the case that $(\left|\braket{\psi|\phi}\right|,\left|\braket{e|f}\right|)=(1,0)$ or $(\left|\braket{\psi|\phi}\right|,\left|\braket{e|f}\right|)=(0,1)$, an arbitrary unital entanglement breaking channel attains the right-hand side of (\ref{eq:ueb_jc}).

\subsection{General entanglement breaking channels}
From Proposition \ref{pro:geb}, we obtain the following corollary.
\begin{cor}
It holds that
\begin{align}
&\max_{\Psi\in\mathcal{GE}}\left\{\frac{1}{2}p_{m_{0}}(\Psi)+\frac{1}{2}q_{n_{0}}(\Psi)\right\} \nonumber \\
&\leq\frac{1}{2}\left(1+\sqrt{\left|\braket{e|f}\right|^{2}+\max\left\{\left|\braket{\psi|\phi}\right|\left|\braket{e|f}\right|,\sqrt{\left(1-\left|\braket{\psi|\phi}\right|^{2}\right)\left(1-\left|\braket{e|f}\right|^{2}\right)}\right\}^{2}}\right). \label{eq:geb_jc}
\end{align}
\end{cor}
The right-hand side of (\ref{eq:geb_jc}) is generally different from the left-hand side.
For example, when $\left|\braket{\psi|\phi}\right|=1$, the right-hand side of (\ref{eq:geb_jc}) becomes
\begin{equation}
\frac{1}{2}\left(1+\sqrt{2}\left|\braket{e|f}\right|\right),
\end{equation}
which is greater than or equal to unity if $\left|\braket{e|f}\right|\geq1/\sqrt{2}$.
Therefore, (\ref{eq:geb_jc}) is valid only when the right-hand side is strictly less than unity.
In other words, we can state that $(\ket{\psi}, \ket{\phi})$ is not jointly convertible into $(\ket{e},\ket{f})$ within $\mathcal{GE}$ if the right-hand side of (\ref{eq:geb_jc}) is strictly less than unity.

\subsection{General quantum channels}
Based on Proposition \ref{pro:gc}, it is evident that the following corollary holds.
\begin{cor}
It holds that
\begin{align}
&\max_{\Psi\in\mathcal{C}}\left\{\frac{1}{2}p_{m_{0}}(\Psi)+\frac{1}{2}q_{n_{0}}(\Psi)\right\} \nonumber \\
&\leq\frac{1}{2}\left(1+\sqrt{\left|\braket{e|f}\right|^{2}+\left(\left|\braket{\psi|\phi}\right|\left|\braket{e|f}\right|+\sqrt{\left(1-\left|\braket{\psi|\phi}\right|^{2}\right)\left(1-\left|\braket{e|f}\right|^{2}\right)}\right)^{2}}\right). \label{eq:gc_jc}
\end{align}
\end{cor}
As with (\ref{eq:geb_jc}), if $\left|\braket{\psi|\phi}\right|=1$ the right-hand side of (\ref{eq:gc_jc}) becomes
\begin{equation}
\frac{1}{2}\left(1+\sqrt{2}\left|\braket{e|f}\right|\right).
\end{equation}
Therefore, (\ref{eq:gc_jc}) is considered a necessary condition for the joint convertibility within $\mathcal{C}$.
It is important to note that a necessary and sufficient condition is given in Proposition \ref{pro:jc_c_iff}.
Hence, (\ref{eq:gc_jc}) is valuable in that it quantifies joint convertibility by means of the average (squared) fidelity rather than merely being a necessary condition.

For example, let us consider the case where $\left|\braket{\psi|\phi}\right|=1/\sqrt{2}$.
In this case, the value inside the first square root in (\ref{eq:gc_jc}) becomes
\begin{align}
&\left|\braket{e|f}\right|^{2}+\left(\left|\braket{\psi|\phi}\right|\left|\braket{e|f}\right|+\sqrt{\left(1-\left|\braket{\psi|\phi}\right|^{2}\right)\left(1-\left|\braket{e|f}\right|^{2}\right)}\right)^{2} \nonumber \\
&=1+\frac{1}{\sqrt{2}}\sin{\left(2x+\frac{\pi}{4}\right)},
\end{align}
where $\left|\braket{e|f}\right|=\cos{x}$ and $0\leq x\leq\pi/2$.
It is easy to verify that the right-hand side of (\ref{eq:gc_jc}) is strictly less than unity if $3\pi/8<x\leq\pi/2$.
In contrast, from Proposition \ref{pro:jc_c_iff}, $(\ket{\psi},\ket{\phi})$ is not jointly convertible into $(\ket{e},\ket{f})$ within $\mathcal{C}$ if and only if $\pi/4<x\leq\pi/2$.
Thus, (\ref{eq:gc_jc}) fails to quantify joint convertibility for $\pi/4<x\leq3\pi/8$.
However, it provides a nontrivial upper bound on the average (squared) fidelity for $3\pi/8<x\leq\pi/2$.

\subsection{Comparison of joint convertibility}
The goal of this subsection is to compare the upper bounds.
As in Subsection \ref{sec:comp_ub}, we denote the right-hand sides of (\ref{eq:gd_jc}), (\ref{eq:unital_jc}), (\ref{eq:ueb_jc}), (\ref{eq:geb_jc}), and (\ref{eq:gc_jc}) as $\Tilde{C}(\tau,\mathcal{X})$.
Except for $\Tilde{C}(\tau,\mathcal{GE})$ and $\Tilde{C}(\tau,\mathcal{C})$, every $\Tilde{C}(\tau,\mathcal{X})$ is equal to the left hand side.
Therefore, we can state that the joint convertibility within $\mathcal{X}$ is higher than that of $\mathcal{Y}$ if $\Tilde{C}(\tau,\mathcal{X})>\Tilde{C}(\tau,\mathcal{Y})$ for $\mathcal{X}\in\{\mathcal{D},\mathcal{R},\mathcal{UE}\}$.
If $\mathcal{X}\subseteq\mathcal{Y}$, it is evident that the joint convertibility within $\mathcal{X}$ does not exceed that of $\mathcal{Y}$.
Indeed, it holds that $\Tilde{C}(\tau,\mathcal{X})\leq\Tilde{C}(\tau,\mathcal{Y})$ in this case.
Hence, we consider the pairs of subsets:
\begin{enumerate*}[label=(\arabic*)]
\item $(\mathcal{UE},\mathcal{D})$,
\item $(\mathcal{R},\mathcal{D})$, and
\item $(\mathcal{R},\mathcal{GE})$.
\end{enumerate*}
In the following, we analyze the function $\Tilde{C}(\tau,\mathcal{X})-\Tilde{C}(\tau,\mathcal{Y})$ for each pair using two examples.

As a first example, we consider the case $\left|\braket{\psi|\phi}\right|=1/\sqrt{2}$, meaning that the input states are selected from mutually unbiased bases.
In this case, the upper bounds are defined as follows:
\begin{align}
\Tilde{C}(\tau,\mathcal{D})&=\frac{1}{2}\left(1+y\right), \\
\Tilde{C}(\tau,\mathcal{R})&=\frac{1}{2}\left(1+\frac{1}{\sqrt{2}}\left(y+\sqrt{1-y^{2}}\right)\right), \\
\Tilde{C}(\tau,\mathcal{UE})&=\frac{1}{2}\left(1+\frac{1}{\sqrt{2}}\max\left\{y,\sqrt{1-y^{2}}\right\}\right), \\
\Tilde{C}(\tau,\mathcal{GE})&=\frac{1}{2}\left(1+\sqrt{y^{2}+\frac{1}{2}\max\left\{y,\sqrt{1-y^{2}}\right\}^{2}}\right).
\end{align}
Here, we denoted $\left|\braket{e|f}\right|$ as $y$, i.e., $y=\left|\braket{e|f}\right|\in[0,1]$.

(1) $(\mathcal{UE},\mathcal{D})$.
The difference between $\Tilde{C}(\tau,\mathcal{UE})$ and $\Tilde{C}(\tau,\mathcal{D})$ is given as
\begin{align}
\Tilde{C}(\tau,\mathcal{UE})-\Tilde{C}(\tau,\mathcal{D})
&=
\begin{cases}
\frac{1}{2}\left(\frac{1}{\sqrt{2}}\sqrt{1-y^{2}}-y\right) & (0\leq y\leq\frac{1}{\sqrt{2}}) \\
\frac{1}{2}\left(\frac{1}{\sqrt{2}}-1\right)y & (\frac{1}{\sqrt{2}}<y\leq1)
\end{cases}
.
\end{align}
The right-hand side is a strictly decreasing function and is equal to zero if and only if $y=1/\sqrt{3}=:y_{1}$.
Consequently, the joint convertibility within $\mathcal{UE}$ is higher than that of $\mathcal{D}$ if and only if $0\leq y<y_{1}$.
Conversely, the joint convertibility within $\mathcal{D}$ is higher than that of $\mathcal{UE}$ if and only if $y_{1}<y\leq1$.
Both joint convertibility values coincide if and only if $y=y_{1}$.

(2) $(\mathcal{R},\mathcal{D})$.
The difference between $\Tilde{C}(\tau,\mathcal{R})$ and $\Tilde{C}(\tau,\mathcal{D})$ is represented as
\begin{align}
\Tilde{C}(\tau,\mathcal{R})-\Tilde{C}(\tau,\mathcal{D})
&=\frac{1}{2}\left(\frac{1}{\sqrt{2}}\left(y+\sqrt{1-y^{2}}\right)-y\right).
\end{align}
Similar to case (1), the right-hand side is a strictly decreasing function and is equal to zero if and only if $y=1/\sqrt{4-2\sqrt{2}}=:y_{2}$.
Thus, the joint convertibility within $\mathcal{R}$ is higher than that of $\mathcal{D}$ if and only if $0\leq y<y_{2}$.
Conversely, the joint convertibility within $\mathcal{D}$ is higher than that of $\mathcal{R}$ if and only if $y_{2}<y\leq1$.
Both joint convertibility values coincide if and only if $y=y_{2}$.

(3) $(\mathcal{R},\mathcal{GE})$.
The difference between $\Tilde{C}(\tau,\mathcal{R})$ and $\Tilde{C}(\tau,\mathcal{GE})$ is represented as
\begin{align}
\Tilde{C}(\tau,\mathcal{R})-\Tilde{C}(\tau,\mathcal{GE})
&=
\begin{cases}
\frac{1}{2\sqrt{2}}\left(y+\sqrt{1-y^{2}}-\sqrt{1+y^{2}}\right) & (0\leq y\leq\frac{1}{\sqrt{2}}) \\
\frac{1}{2\sqrt{2}}\left((1-\sqrt{3})y+\sqrt{1-y^{2}}\right) & (\frac{1}{\sqrt{2}}<y\leq1)
\end{cases}
.
\end{align}
The right-hand side is equal to zero if and only if $y=0$ or $y=1/\sqrt{5-2\sqrt{3}}=:y_{3}$, and is positive if $0<y<y_{3}$.
Consequently, the joint convertibility within $\mathcal{R}$ is higher than that of $\mathcal{GE}$ if $0<y<y_{3}$.
It should be noted that $\Tilde{C}(\tau,\mathcal{GE})$ is not the optimal value, but an estimate.
Therefore, no conclusions can be drawn when the right-hand side is less than or equal to zero.

The second example pertains to the case where $\left|\braket{e|f}\right|=1/\sqrt{2}$, meaning that the two measurements on the output are conducted using mutually unbiased bases.
In this case, the upper bounds are represented as follows:
\begin{align}
\Tilde{C}(\tau,\mathcal{D})&=\frac{1}{2}\left(1+\frac{1}{\sqrt{2}}\right), \\
\Tilde{C}(\tau,\mathcal{R})&=\frac{1}{2}\left(1+\frac{1}{\sqrt{2}}\left(x+\sqrt{1-x^{2}}\right)\right), \\
\Tilde{C}(\tau,\mathcal{UE})&=\frac{1}{2}\left(1+\frac{1}{\sqrt{2}}\max\left\{x,\sqrt{1-x^{2}}\right\}\right), \\
\Tilde{C}(\tau,\mathcal{GE})&=\frac{1}{2}\left(1+\frac{1}{\sqrt{2}}\sqrt{1+\max\left\{x,\sqrt{1-x^{2}}\right\}^{2}}\right).
\end{align}
Here, we denoted $\left|\braket{\psi|\phi}\right|$ as $x$; i.e., $x=\left|\braket{\psi|\phi}\right|\in[0,1]$.

(1) $(\mathcal{UE},\mathcal{D})$.
The difference between $\Tilde{C}(\tau,\mathcal{UE})$ and $\Tilde{C}(\tau,\mathcal{D})$ is given by
\begin{align}
\Tilde{C}(\tau,\mathcal{UE})-\Tilde{C}(\tau,\mathcal{D})
&=
\begin{cases}
\frac{1}{2\sqrt{2}}\left(\sqrt{1-x^{2}}-1\right) & (0\leq x\leq1/\sqrt{2}) \\
\frac{1}{2\sqrt{2}}\left(x-1\right) & (1/\sqrt{2}<x\leq1)
\end{cases}
.
\end{align}
The right-hand side is zero if and only if $x\in\{0,1\}$ and negative otherwise.
Hence, the joint convertibility within $\mathcal{D}$ is higher than that of $\mathcal{UE}$ if and only if $x\in(0,1)$; otherwise, both coincide.

(2) $(\mathcal{R},\mathcal{D})$.
The difference between $\Tilde{C}(\tau,\mathcal{R})$ and $\Tilde{C}(\tau,\mathcal{D})$ is represented as
\begin{align}
\Tilde{C}(\tau,\mathcal{R})-\Tilde{C}(\tau,\mathcal{D})
&=\frac{1}{2\sqrt{2}}\left(x+\sqrt{1-x^{2}}-1\right).
\end{align}
The right-hand side is zero if and only if $x\in\{0,1\}$ and positive otherwise.
Consequently, the joint convertibility within $\mathcal{R}$ is higher than that of $\mathcal{D}$ if and only if $x\in(0,1)$; otherwise both coincide.

(3) $(\mathcal{R},\mathcal{GE})$.
The difference between $\Tilde{C}(\tau,\mathcal{R})$ and $\Tilde{C}(\tau,\mathcal{GE})$ is given by
\begin{align}
\Tilde{C}(\tau,\mathcal{R})-\Tilde{C}(\tau,\mathcal{GE})
&=
\begin{cases}
\frac{1}{2\sqrt{2}}\left(x+\sqrt{1-x^{2}}-\sqrt{2-x^{2}}\right) & (0\leq x\leq1/\sqrt{2}) \\
\frac{1}{2\sqrt{2}}\left(x+\sqrt{1-x^{2}}-\sqrt{1+x^{2}}\right) & (1/\sqrt{2}<x\leq1)
\end{cases}
.
\end{align}
The right-hand side is positive if and only if $x\in(1/\sqrt{5},2/\sqrt{5})$.
Consequently, the joint convertibility within $\mathcal{R}$ is higher than that of $\mathcal{GE}$ if $x\in(1/\sqrt{5},2/\sqrt{5})$.
As with the previous example, no conclusions can be drawn when the right-hand side is less than or equal to zero. 

Hence, joint convertibility can be compared analytically using the upper bounds.

While (\ref{eq:geb_jc}) and (\ref{eq:gc_jc}) are estimates, akin to Proposition \ref{pro:bound_c_sp}, they achieve equality under the following special case.
\begin{pro} \label{pro:jc_sp}
When $\left|\braket{e|f}\right|=0$, the equality
\begin{equation}
\max_{\Psi\in\mathcal{X}}\left\{\frac{1}{2}p_{m_{0}}(\Psi)+\frac{1}{2}q_{n_{0}}(\Psi)\right\}
=\frac{1}{2}\left(1+\sqrt{1-\left|\braket{\psi|\phi}\right|^{2}}\right) \label{eq:c_jc_sp}
\end{equation}
holds for all $\mathcal{X}\in\{\mathcal{R},\mathcal{UE},\mathcal{GE},\mathcal{C}\}$
\end{pro}
\begin{proof}
It is straightforward to verify that (\ref{eq:c_jc_sp}) holds for $\mathcal{R}$ and $\mathcal{UE}$ because (\ref{eq:unital_jc}) and (\ref{eq:ueb_jc}) are equality.

From (\ref{eq:geb_jc}) and (\ref{eq:gc_jc}), it follows that
\begin{equation}
\max_{\Psi\in\mathcal{X}}\left\{\frac{1}{2}p_{m_{0}}(\Psi)+\frac{1}{2}q_{n_{0}}(\Psi)\right\}
\leq\frac{1}{2}\left(1+\sqrt{1-\left|\braket{\psi|\phi}\right|^{2}}\right) \label{eq:c_jc_ineq}
\end{equation}
for all $\mathcal{X}\in\{\mathcal{GE},\mathcal{C}\}$.
From $\mathcal{UE}\subseteq\mathcal{X}$, the inequality
\begin{equation}
\max_{\Psi\in\mathcal{UE}}\left\{\frac{1}{2}p_{m_{0}}(\Psi)+\frac{1}{2}q_{n_{0}}(\Psi)\right\}
\leq\max_{\Psi\in\mathcal{X}}\left\{\frac{1}{2}p_{m_{0}}(\Psi)+\frac{1}{2}q_{n_{0}}(\Psi)\right\}
\end{equation}
also holds, and the left-hand side is equal to the right-hand side of (\ref{eq:c_jc_ineq}).
Thus, (\ref{eq:c_jc_ineq}) must be equality.
\end{proof}
This proposition states that if $\left|\braket{e|f}\right|=0$, then the notion of joint convertibility for $\mathcal{R},\mathcal{UE},\mathcal{GE},$ and $\mathcal{C}$ are all equivalent.
Specificlly, when $\left|\braket{e|f}\right|=0$, $(\ket{\psi},\ket{\phi})$ is jointly convertible within $\mathcal{X}$ if and only if $\left|\braket{\psi|\phi}\right|=0$ for all $\mathcal{X}\in\{\mathcal{R},\mathcal{UE},\mathcal{GE},\mathcal{C}\}$.
The right-hand side of (\ref{eq:c_jc_sp}) is achieved by (\ref{eq:unital_max}) and (\ref{eq:ueb_max}) for $\{\mathcal{R},\mathcal{C}\}$ and $\{\mathcal{R},\mathcal{UE},\mathcal{GE},\mathcal{C}\}$, respectively.

\section{Detection of a quantum channel} \label{sec:id}
Let us consider the unital quantum channel
\begin{align}
\Psi_{w}=w\id+(1-w)\Psi_{I/2}, \label{eq:id_sample}
\end{align}
where $0\leq w\leq1$ and $\Psi_{I/2}$ represents the depolarizing channel that outputs $I/2$.
It is evident that while $\Psi_{w=1}$ is a unitary channel, $\Psi_{w=0}$ is an entanglement breaking channel.
Hence, a threshold value may be exist that divides the set $\Set{\Psi_{w}|0\leq w\leq1}$ into entanglement breaking channels and not entanglement breaking ones.
To determine this value, we consider two contrasting measurements.

\subsection{Measurement with the maximally entangled state}
The first measurement is given by the PPOVM
\begin{equation}
\left\{\frac{1}{2}P_{+},\frac{1}{2}\left(I-P_{+}\right)\right\}, \label{eq:test_1_ppovm}
\end{equation}
which is executed using the maximally entangled state $P_{+}$ as the input state and the binary measurement $\left\{P_{+},I-P_{+}\right\}$ as the POVM measurement on the output.
Consequently, this measurement actively employs entanglement as a resource.

Let us assume that $\Psi_{w}$ is measured by (\ref{eq:test_1_ppovm}).
The probability of obtaining the outcome associated with $P_{+}/2$ is expressed as
\begin{equation}
\frac{1}{2}\tr\left[P_{+}J_{\Psi_{w}}\right]
=\frac{1}{4}\left(1+3w\right).
\end{equation}
From Proposition \ref{pro:ueb}, if $\Psi_{w}$ is an entanglement breaking channel, this value is bounded from above as
\begin{equation}
\frac{1}{4}\left(1+3w\right)
\leq\frac{1}{4}\left(1+\left\|N(P_{+})\right\|\right).
\end{equation}
It is easy evident that $\left\|N(P_{+})\right\|=1$ from (\ref{eq:P+_mat}).
Hence, this inequality becomes
\begin{equation}
\frac{1}{4}\left(1+3w\right)
\leq\frac{1}{2},
\end{equation}
and it is equivalent to
\begin{equation}
w\leq\frac{1}{3},
\end{equation}
which indicates that $\Psi_{w}$ is not entanglement breaking when $w>1/3$.
Note that a necessary and sufficient condition for $\Psi_{w}$ to be entanglement breaking is that $(\id\otimes\Psi_{w})P_{+}$ becomes a separable state\cite{doi:10.1142/S0129055X03001709}.
Consequently, $\Psi_{w}$ is not entanglement breaking if and only if $(\id\otimes\Psi_{w})P_{+}$ is entangled.
The operator $(\id\otimes\Psi_{w})P_{+}$ is known as the Werner state, and it has been already shown that the Werner state is entangled if and only if $w>1/3$\cite{PITTENGER2000447}.
Thus, the criterion provided by (\ref{eq:bound_ue}) is so strong as to detect all $\Psi_{w}$ that are not entanglement breaking channels.

The purpose of this measurement is to determine whether $\Psi_{w}$ is not entanglement breaking.
Therefore, this result can be considered reasonable, as the maximally entangled state $P_{+}$ is likely to be highly sensitive to the action of entanglement breaking channels.

\subsection{Ancilla-free measurements}
The second measurement is given by the three PPOVMs
\begin{align}
&\{(\ket{0}\bra{0})^{T}\otimes\ket{0}\bra{0},(\ket{0}\bra{0})^{T}\otimes\ket{1}\bra{1}\}, \label{eq:test_2_ppovm_1} \\
&\{(\ket{+}\bra{+})^{T}\otimes\ket{+}\bra{+},(\ket{+}\bra{+})^{T}\otimes\ket{-}\bra{-}\}, \label{eq:test_2_ppovm_2} \\
&\{(\ket{+_{y}}\bra{+_{y}})^{T}\otimes\ket{+_{y}}\bra{+_{y}},(\ket{+_{y}}\bra{+_{y}})^{T}\otimes\ket{-_{y}}\bra{-_{y}}\}, \label{eq:test_2_ppovm_3}
\end{align}
where $\ket{\pm}$ and $\ket{\pm_{y}}$ are the eigenvectors of $\sigma_{1}$ and $\sigma_{2}$ defined by
\begin{equation}
\ket{\pm}=(1/\sqrt{2})(\ket{0}\pm\ket{1})
\end{equation}
and
\begin{equation}
\ket{\pm_{y}}=(1/\sqrt{2})(\ket{0}\pm i\ket{1}),
\end{equation}
respectively.
As (\ref{eq:test_2_ppovm_1}), (\ref{eq:test_2_ppovm_2}), and (\ref{eq:test_2_ppovm_3}) are ancilla-free measurements, unlike (\ref{eq:test_1_ppovm}), these PPOVMs do not use entanglement.

If $\Psi_{w}$ is randomly measured using these PPOVMs, the probability of obtaining one of the outcomes $0$, $+$, and $+_{y}$ is given by
\begin{align}
\tr\left[\tau J_{\Psi_{w}}\right]
=\frac{1+w}{2}, \label{eq:prob_werner}
\end{align}
where $\tau$ is defined by
\begin{align}
\tau=\frac{1}{3}\left((\ket{0}\bra{0})^{T}\otimes\ket{0}\bra{0}+(\ket{+}\bra{+})^{T}\otimes\ket{+}\bra{+}+(\ket{+_{y}}\bra{+_{y}})^{T}\otimes\ket{+_{y}}\bra{+_{y}}\right).
\end{align}
From (\ref{eq:bound_ue}), if $\Psi_{w}$ is an entanglement breaking channel, (\ref{eq:prob_werner}) is bounded from above as follows:
\begin{align}
\frac{1+w}{2}
\leq\frac{1}{2}\left(1+\|N(\tau)\|\right). \label{eq:cri_werner}
\end{align}
An arbitrary element of $N(\tau)$ can be expressed as
\begin{align}
N(\tau)_{ij}
=\frac{1}{3}(\delta_{i1}\delta_{j1}-\delta_{i2}\delta_{j2}+\delta_{i3}\delta_{j3}).
\end{align}
From $N(\tau)N(\tau)^{T}=I/9$ it holds that $\|N(\tau)\|=1/3$.
Thus, (\ref{eq:cri_werner}) becomes
\begin{align}
\frac{1+w}{2}\leq\frac{1}{2}\left(1+\frac{1}{3}\right)
\end{align}
which is equivalent to
\begin{align}
w\leq\frac{1}{3}.
\end{align}
That is, if $w>1/3$, $\Psi_{w}$ is not an entanglement breaking channel.
Consequently, the criterion provided by (\ref{eq:test_2_ppovm_1}), (\ref{eq:test_2_ppovm_2}), and (\ref{eq:test_2_ppovm_3}) coincides with that offered by (\ref{eq:test_1_ppovm}), detecting all $\Psi_{w}$ that are not entanglement breaking channels.

It is important to note that although this measurement requires three different PPOVMs, it does not rely on entanglement to confirm that $\Psi_{w}$ is not entanglement breaking.
In general, experimentally generating an entangled state is not easy.
Therefore, our detection method may serve as a valuable criterion for determining whether a given unital quantum channel is not entanglement breaking.

\section{Conclusions} \label{sec:conc}
In this study, we derived the several upper bounds on the probabilities in measurements of qubit channels.
These bounds are expressed without requiring an optimization problem, simplifying their calculations.

As an application, we used the derived upper bounds to quantify the notion of convertibility, which we call joint convertibility, and obtained various nontrivial results.
Consequently, the conditions for joint convertibility were examined in detail.
In addition, it was illustrated that the joint convertibility for different classes of quantum channels can be compared using these upper bounds.

As another application, we considered the problem of detecting a property of a simple unital channel with a single parameter.
Specifically, it was demonstrated that one can confirm whether the unital quantum channel is not entanglement breaking by utilizing one of the derived upper bounds.
To apply this upper bound, we considered two contrasting measurements.
The criteria provided by these measurements are strong enough to give the necessary and sufficient conditions that the unital channel is not entanglement breaking.

\section*{Acknowledgment}
TM acknowledges financial support from JSPS (KAKENHI Grant No. JP20K03732).

\section*{Appendix: The proof of Lemma \ref{lem:sing_inv}}
For simplicity, let us define $\rho'=(V_{1}\otimes V_{2})\rho(V^{\dagger}_{1}\otimes V^{\dagger}_{2})$.
From (\ref{eq:def_corr_mat}), it follows that
\begin{align}
N(\rho')_{ij}
&=\tr\left[(V_{1}\otimes V_{2})\rho(V^{\dagger}_{1}\otimes V^{\dagger}_{2})(\sigma_{i}\otimes\sigma_{j})\right] \nonumber \\
&=\tr\left[\rho(V^{\dagger}_{1}\sigma_{i}V_{1}\otimes V^{\dagger}_{2}\sigma_{j}V_{2})\right]. \label{eq:loc_tra_rho}
\end{align}
We can expand $V^{\dagger}_{1}\sigma_{i}V_{1}$ and $V^{\dagger}_{2}\sigma_{j}V_{2}$ as
\begin{align}
V^{\dagger}_{1}\sigma_{i}V_{1}&=\sum^{3}_{k=1}O(V^{\dagger}_{1})_{ki}\sigma_{k}, \\
V^{\dagger}_{2}\sigma_{j}V_{2}&=\sum^{3}_{l=1}O(V^{\dagger}_{2})_{lj}\sigma_{l},
\end{align}
where $O(V^{\dagger}_{1})_{ki}=(1/2)\tr[\sigma_{k}V^{\dagger}_{1}\sigma_{i}V_{1}]$ and $O(V^{\dagger}_{2})_{lj}=(1/2)\tr[\sigma_{l}V^{\dagger}_{2}\sigma_{j}V_{2}]$ are orthonormal matrices.
By substituting these expressions into (\ref{eq:loc_tra_rho}), we obtain
\begin{align}
N(\rho')_{ij}
&=(O(V^{\dagger}_{1})^{T}N(\rho)O(V^{\dagger}_{2}))_{ij},
\end{align}
which means
\begin{align}
N(\rho')=O(V^{\dagger}_{1})^{T}N(\rho)O(V^{\dagger}_{2}).
\end{align}
Let $\lambda$ be an arbitrary eigenvalue of $N(\rho)^{T}N(\rho)$, and $\mathbf{x}$ be the corresponding eigenvector.
From
\begin{align}
N(\rho')^{T}N(\rho')O(V^{\dagger}_{2})^{T}\mathbf{x}
&=\lambda O(V^{\dagger}_{2})^{T}\mathbf{x},
\end{align}
it is evident that $\lambda$ is also an eigenvalue of $N(\rho')^{T}N(\rho')$, and $O(V^{\dagger}_{2})^{T}\mathbf{x}$ is the corresponding eigenvector.
This demonstrates that $N(\rho')^{T}N(\rho')$ and $N(\rho)^{T}N(\rho)$ share the same eigenvalues, as $\lambda$ is chosen as an arbitrary eigenvalue.
Therefore, it is proven that $N(\rho')$ and $N(\rho)$ have the same singular values.

\end{document}